\documentclass[12pt]{article}
\usepackage{mathpazo}
\usepackage{natbib}

\usepackage[utf8]{inputenc}
\usepackage{graphicx}
\usepackage[dvipsnames]{xcolor}

\usepackage{amsmath}
\usepackage{amsthm}
\usepackage{amsfonts}
\usepackage{amssymb}
\usepackage{enumerate}
\usepackage{comment}

\usepackage{epstopdf}
\usepackage{geometry}
\usepackage{ulem}
\usepackage{soul}
\usepackage[colorlinks=true,linkcolor=blue,urlcolor=blue,citecolor=blue, hyperfigures=false]{hyperref}

\usepackage{tikz}
\usetikzlibrary{decorations.pathreplacing,calligraphy}

\definecolor{PineGreen}{cmyk}{0.9,0,1,0.40}

\definecolor{ProcessBlue}{cmyk}{1,0,0,0.40}
\newtheorem{theorem}{Theorem}
\newtheorem*{theorem*}{Theorem}

\newtheorem{lemma}{Lemma}

\newtheorem{assumption}{Assumption}

\newtheorem{remark}{Remark}
\newtheorem{conclusion}{Conclusion}

\newcommand{\dR}{\mathbb R}

\newcommand{\E}{\mathbb E}

\def \ep{\varepsilon}

\newcommand{\cav}{{\rm cav\ }}
\newcommand{\prob}{{\mathbb P}}

\newcommand{\rmd}{{\rm d}}

\newtheorem{Theorem}{Theorem}[section]

\newcounter{figurecounter}
\usepackage{algorithm}
\usepackage{algorithmic}
\def\undertilde#1{\mathord{\vtop{\ialign{##\crcr
$\hfil\displaystyle{#1}\hfil$\crcr\noalign{\kern1.5pt\nointerlineskip}
$\hfil\widetilde{}\hfil$\crcr\noalign{\kern1.5pt}}}}}

\begin{document}
\title{\textbf{Markovian Persuasion with Two States}\thanks{Ashkenazi-Golan acknowledges the support of the Israel Science Foundation, Grants 217/17 and 722/18, and NSFC-ISF Grant 2510/17. Hern\'andez acknowledges the support of the Consellería d’Innovació, Universitats, Ciència i Societat Digital, Generalitat Valenciana, grant number AICO/2021/257, and the Ministerio de Ciencia, grant number PID2021-128228NB-I00. Neeman acknowledges the support of the Israel Science Foundation, Grant 1465/18. Solan acknowledges the support of the Israel Science Foundation, Grant 217/17.}}
\author{Galit Ashkenazi-Golan\footnote{London School of Economics and Political Science,  \tt{galit.ashkenazi@gmail.com}.} \and Pen\'elope Hern\'andez\footnote{ERI-CES, Department of Economics, University of Valencia, \tt{ penelope.hernandez@uv.es}.} \and Zvika Neeman\footnote{Berglas School of Economics, Tel-Aviv University, \tt{zvika@tauex.tau.ac.il.}} \and Eilon Solan\footnote{The School of Mathematical Sciences, Tel Aviv University, \tt{eilons@tauex.tau.ac.il}.}}
\date{\today}
\maketitle

\begin{abstract}
This paper addresses the question of how to best communicate information over time in order to influence an agent's belief and induced actions in a model with a  binary state of the world that evolves according to a Markov process, and with a finite number of actions.  
We characterize the sender's optimal message strategy in the limit, as the length of each period decreases to zero. The optimal strategy 
is not myopic. Depending on the agent's beliefs, sometimes no information is revealed, and sometimes the agent's belief is split into two well-chosen posterior beliefs.
%involves sliding (silence) and binary splits; it is not myopic.  


\bigskip
\end{abstract}
\thispagestyle{empty}

\newpage
\setcounter{page}{1}

\section{Introduction}

This paper addresses the question of how to best communicate information over time in order to influence an agent's beliefs and induced actions.
We consider a model in which a binary state of the world evolves according to a Markov process. In every period, a sender (she) observes the state and sends a message to a myopic receiver (he). 
The message that is sent by the sender induces the myopic receiver's belief and so action in that period, but also affects the future beliefs of the receiver, and so also the way in which the receiver would respond to future messages. The question is how should the sender balance current and future implications of her messages.

%Our goal is to characterize the optimal sender's strategy, when the time length between periods is small. Our main theoretical contribution is the development of methods for the analysis of dynamic Bayesian persuasion models with \textit{discontinuous payoffs}. 

\cite{ely_beeps_2017} and \cite{Renault_2017} have studied such models (we provide a detailed discussion of their work below), 
and have characterized the sender´s optimal strategy when the receiver has only two actions.
They showed that in this case the sender's optimal strategy is myopic. That is, the sender´s optimal policy ignores the effect of the sender´s messages on the receiver´s future beliefs.
In contrast, we allow for any finite number of actions,
and find that the larger set of actions calls for a non-myopic sender's optimal strategy.

%\red{\st{To this end, we study the continuous-time limit of the game, identify the optimal sender's strategy in this limit game, and show that this strategy is also approximately optimal when the time length between periods is small.}} 

%We assume that the receiver's set of available actions is finite. 
For simplicity, we assume that the receiver's action is an increasing function of his belief, and the sender's payoff is an increasing function of the receiver's action. 
This assumption implies that the receiver's optimal strategy is piecewise-constant in his beliefs. That is, the space of beliefs, which is represented by the unit interval, can be divided into finitely many subintervals,
and in each such subinterval the optimal strategy of the receiver is a fixed action.
This in turn implies that the sender's stage payoff, which is a function of the receiver's belief and the receiver's action, is discontinuous in the receiver's belief. We focus on the case in which the sender's indirect payoff, as a function of the receiver's beliefs about the state, has a concave envelope. Such a concave envelope arises naturally when the sender's marginal benefit from the receiver's action is decreasing.
%We consider the continuous time limit of this model where the length of each period decreases to zero, and the transition time between the two states is distributed according to an exponential distribution. 

For example, consider a seller of an experience good such as wine, whose quality changes stochastically over time depending on local climate. The seller prefers that buyers buy as much wine as possible, but obtains a decreasing marginal benefit from each case sold. The seller may disclose information about the wine quality to buyers who decide what quantity of wine to purchase in each period. 
%Or consider a meteorologist who observes weather conditions, which are either hotter or colder than average in every period. The meteorologist cares about global warming and so prefers that firms in the economy produce less so as to forestall global warming but obtains a smaller marginal benefit from additional reductions in the quantity produced, reports weather conditions to firms who decide how much to produce.

The main conceptual contribution of the paper is the understanding of the driving forces behind the sender's optimal strategy. Standard results imply that, in every period, the sender can induce any distribution of posterior beliefs whose mean is equal to the belief in the previous period plus the one-period drift in the Markov process. Such distributions are said to be ``Bayes plausible'' (Kamenica and Gentzkow, 2011). We show that the optimal strategy for the sender involves only two types of distributions of induced beliefs. The first distribution arises as a consequence of the sender's silence. 
In this case, the receiver's belief ``slides'' 
toward the invariant distribution of the Markov process. The second distribution requires simple communication and consists of a binary split of the receiver's posterior belief. 

Suppose that $p < p'$ are two beliefs that lie below the invariant distribution of the Markov process.
Suppose that the current receiver's belief is $p$. The dynamics pushes the
current
belief toward $p'$.
The observation mentioned in the previous paragraph suggests two strategies that the sender can use to facilitate this change in the receiver's beliefs:
(a) the sender reveals no information until the belief becomes $p'$ (``silence''),
and (b) the sender repeatedly splits the receiver's belief between $p$ and $p'$, until the belief finally coincides with $p'$.
It turns out that when comparing the discounted time it takes the belief to reach $p'$ under these two strategies,
the latter strategy is quicker.
A similar result holds when $p > p'$ and the two beliefs lie above the invariant distribution of the Markov process.

Since the sender's payoff is monotone in the receiver's belief,
the sender would like the belief to be as high as possible.
This ``speed-based'' argument suggests that when the current receiver's belief is below the invariant distribution, repeated splitting would be better for the sender because it is quicker in generating receiver's beliefs that are more favorable for the sender; and when the current receiver's belief is above the invariant distribution, silence would be better for the sender because it is slower in generating receiver's beliefs that are less favorable for the sender.

However, these two strategies also generate different instantaneous payoffs for the sender:
if the sender repeatedly splits the receiver's belief between $p$ and $p'$, then her instantaneous payoff is a weighted average of her instantaneous payoffs at $p$ and $p'$; and if the sender reveals no information, then her instantaneous payoff is the instantaneous payoff as the beliefs slide from $p$ to $p'$.
The sender's instantaneous payoff is increasing discontinuously in the receiver's beliefs. Payoffs to the left of a discontinuity point are significantly smaller than the payoff at the discontinuity point.
This ``payoff-based'' argument suggests that repeated splitting yields higher instantaneous payoffs than no revelation of information both below and above the invariant distribution of the Markov process.

For receiver's beliefs that lie below the invariant distribution, the speed-based and payoff-based forces are in agreement, and so for such receiver's beliefs, the sender's optimal strategy involves repeated splitting of the receiver's belief between the discontinuity points of the sender's payoff function, and it is myopic.
But for receiver's beliefs that lie above the invariant distribution, the speed-based and payoff-based forces work in opposite directions.
We show that for such beliefs the sender's optimal strategy is not myopic:
at beliefs that are slightly above a discontinuity point, the difference between the instantaneous payoffs of the two strategies is small, the speed-based argument dominates, and the sender reveals no information;
while at beliefs that are slightly below a discontinuity point, the difference between the instantaneous payoffs of the two strategies is large, the speed-based argument dominates, and 
the sender repeatedly splits the receiver's belief between the belief at the discontinuity point and a well-chosen belief below it.

In the context of the example of the wine seller, when buyers believe that wine quality is higher than average, then the seller´s optimal strategy is myopic. In this case, the seller need not worry about the buyers´ future beliefs. But when buyers believe that wine quality is lower than average then the seller´s optimal strategy is not myopic. In this case the seller´s optimal policy is more elaborate and it involves both silence and splitting of the buyers´ beliefs. 

%\red{Finally, it is important to note that we concede that the monotonicity and concave envelope assumptions are indeed strong. However, they are very useful in that they allow us to solve the problem, and obtain a solution that can be relatively easily described in terms of the parameters of the problem, and the properties of the sender's indirect payoff function. More general problems can be divided into components, such that the sender's indirect payoff function in each component is monotone and with a concave or convex envelope. Each such component can be solved using the methods introduced in this paper, and the solution of the general problem can be obtained by connecting the specific solutions of the different components. For simplicity, we have chosen not to incorporate the solution of this more general problem into the paper.}

\subsection*{Related Literature}
Our work relates to two distinct research directions that should probably be more closely linked to each other: the one on information design and Bayesian persuasion, and the second on repeated games under incomplete information. For recent surveys of these two directions, see \cite{kamenica_bayesian_2019} and \cite{Mertens_2016}, respectively.

The model studied in this paper is a dynamic extension of the static Bayesian persuasion model of \cite{kamenica_bayesian_2011}. Our model is a generalization of the model studied by \cite{ely_beeps_2017}, who solved a simpler version of our model with binary actions and states, where one of the two states is absorbing. \cite{ely_beeps_2017} showed that in this case, the optimal strategy is to reveal the (absorbing) state with delay. A delay of time $T$ implies that, starting at time zero, the receiver's beliefs that the state has switched slides upwards with time. Denote the belief that the state is absorbing at time $t$ by $p^t$. Under the sender's optimal strategy, at time $T$, the receiver still hasn't learned anything about the state, and so his belief $p^T$ reflects the knowledge that no switch has occurred until at least $T$ moments ago. For $t>T$, the fact that the sender reveals that the state has switched with delay $T$ implies that the receiver's beliefs are either given by $p^t=1$ or $p^t=p^T$, because a receiver who is told that the state has switched updates his belief to $p^t=1$, and a receiver who hasn't heard anything yet knows only that the state was not absorbing $T$ moments ago, and so has beliefs $p^t=p^T$. Thus, the optimal strategy identified by \cite{ely_beeps_2017} also combines sliding and splitting between the two beliefs $p=p^T$ and $p=1$.

Silence in our model can also be interpreted as delay in Ely's (\citeyear{ely_beeps_2017}) model. But unlike in Ely's model, the silences that are part of the optimal strategy identified here vary in length, and are punctuated by messages that induce beliefs that reflect different levels of certainty about the current state. 

Another related paper is that of \cite{Renault_2017}, who consider a similar model, again with just two actions. They show that with two states, a greedy or myopic strategy is optimal for the sender (the optimal strategy in \citeauthor{ely_beeps_2017}'s (\citeyear{ely_beeps_2017}) model is myopic as well). 
As explained above, this is not the case here. \citeauthor{Renault_2017}  also show that with more than two states, the optimal strategy for the sender need not be myopic, and provide a sufficient condition on the Markov dynamic process that ensures it is myopic.

Recent papers by \cite{ball_dynamic_2019}, \cite{ely_moving_2020}, and
\cite{smolin_dynamic_2017} obtain results about the timing of the optimal revelation of information in specific settings. Their focus is more on the optimal time to reveal information, rather than what information to reveal, as in this paper. There is also a large literature in economics on the design of information feedback in dynamic principal-agent problems and games (see, e.g., the literature review in \cite{ely_beeps_2017}). However, as noted by Ely, with a few exceptions, these papers generally consider exogenous information structures or compare a few policies such as full, public, and no disclosure.

The two key methods that are used in the Bayesian persuasion literature are Bayes plausibility and the geometric characterization of the optimum through a concavification of the sender's indirect payoff function. Both of these ideas were adapted from the work of \cite{aumann_repeated_1995}, who studied repeated games with one-sided incomplete information. In \cite{aumann_repeated_1995}, one of two players learns which of two two-player 0-sum normal form stage games is to be played, and then this game is played repeatedly. Their analysis has been extended by \cite{renault_value_2006} to cover Markov games. In the setting with incomplete information studied by \cite{aumann_repeated_1995}, the sender reveals information only once, in the first period of the game, and then continues to play in a way that is uninformative for the rest of the game. In contrast,  in the Markov game studied by \cite{renault_value_2006}, the state of the world changes over time, and so it is optimal for the sender to continue to reveal information about the state as it evolves. \cite{cardaliaguet_markov_2016} and \cite{gensbittel_continuous-time_2016} show that the value of a 
continuous-time
Markov game is given by the solution of a differential equation (but stop short of obtaining explicit solutions). \cite{ashkenazi-golan_solving_2020} present an algorithm that converges to a solution of this differential equation, but they only apply it to a few examples of two-player zero-sum two-state Markov games with one-sided information. They provide the important insight that it is possible to characterize the optimal strategy through the value of the derivative of the putative value function given a specific suggested split of the uninformed player's beliefs. 

As mentioned above, our assumptions about structure of the family of the sender's indirect payoff function allows us to obtain an explicit characterization of the optimal information strategy for the sender in a class of dynamic Bayesian persuasion problems. In addition, we also compute the expected discounted time it takes to switch from one induced posterior belief to another, which is relevant also to two-player Markov games. We rely on similar ideas to those in \cite{ashkenazi-golan_solving_2020}
to obtain an explicit description of the solution for a class of dynamic persuasion games. However, while the entire literature on repeated games under incomplete information has restricted its attention to the special case where the informed player's indirect payoff function is \textit{continuous} in the uninformed player's beliefs, we study \textit{discontinuous} indirect payoff functions. And the finiteness of the receiver's set of actions assumed here necessarily implies that the sender's indirect payoff function is indeed discontinuous.\footnote{Moreover, even if the receiver has an infinite number of actions available, strong continuity requirements need to be imposed on the payoffs to ensure that the sender's indirect payoff function would be continuous.}
\bigskip

The rest of the paper proceeds as follows. In Section 2, we present the model. In Section 3, we present our main results and discuss possible extensions. All proofs are relegated to the Appendix.

\section{Model and Main Result}\label{sec model}

We consider a 
discrete-time
game 
with two players: a sender (she) and a receiver (he). In every period $n\in\{1,2,\ldots\}$, the sender observes the state of the world $\omega^n\in\{0,1\}$ and sends a message $m^n\in M$ to the receiver, who takes an action $a^n\in A$. The set $A$ is assumed to be finite
and the set $M$ contains at least two messages. 

\medskip
\noindent\textbf{Markov Transitions.}

The probability that the initial state is 0 is given, and denoted by $p^0$.
%\color{green}/// Zvika this was previously denoted $p^0$, which is not good notation because we later denote $p^t$ in continuous time. So I replaced it with $p^\circ$ here and below in the formulation of the sender's OF

In each period there is a constant probability, which may depend on the current state but is independent of the history of the play
prior to the current period,
that the state switches to the other state in the next period. The switches between the states give rise to a Markov chain. This Markov chain has a stationary distribution. We denote the probability that the state is 1 according to the stationary distribution by $p^*$. Standard results in the theory of Markov chains imply that starting with any initial probability $p^0$, 
$p^n$
converges to $p^*$ as $n$ increases. 

\medskip
\noindent\textbf{Posterior Beliefs.}

We assume that
the sender is committed to her message strategy, and the receiver is aware of this commitment. 
As a result, at the beginning of each period $n$,
the receiver updates his belief $p^n$ that the state is 1 
given the message he received in stage $n-1$ and taking into account the Markov transition. 
%Because there are two states, we can summarize the receiver's belief $p^n$ in period $n$ by the probability he assigns to state $\omega^1$. 

\medskip
\noindent\textbf{Stage Payoffs.} 

In any period $n$, the (period-$n$) payoffs of both the sender and receiver are functions of the state $\omega^n$ and the receiver's action $a^n$ in period $n$. The receiver is assumed to be myopic: in every period $n$, he chooses the action $a^n$ that maximizes his payoff given his belief $p^n$ over the states.\footnote{We assume that when indifferent, the receiver chooses the action that is better for the sender. This assumption is made for convenience only. If, when indifferent, the receiver chooses the action that is less favorable to the sender with a positive probability,
then the sender would ensure that the receiver is never indifferent between these two actions.}
It follows that the sender's (indirect) payoff in period $n$, denoted $u : [0,1] \to \dR$, may be viewed as a function of the receiver's belief in that period, $p^n$.

We assume that the receiver's action is increasing in his belief, and the sender's payoff is increasing in the receiver's action. Monotonicity together with the fact that the set of actions $A$ is finite imply that $u$ is increasing and piecewise constant. We assume, in addition, that the function $u$ has a ``concave envelope.'' 
That is, the piecewise linear function that connects all the discontinuity points of $u$ is concave (see Fig~\ref{fig:M1}): 

\begin{assumption}
\label{assumption:1}
There exist $m,m'\in \mathbb{N}$,
$0=p_{-m}<p_{-m+1}<\dots<p_0<p_1<\dots< p_{m'}=1$,
and $h_{-m} < h_{-m+1} < \dots < h_{m' -1}$ such that $u(p)=h_i$ for $p\in[p_i,p_{i+1})$ if $i\in\{-m,\ldots,m' -1\}$, and $u(1) = h_{m' -1}$. 
Moreover, the line segment that connects the points $(p_{i-1}, h_{i-1})$ and $(p_{i+1}, h_{i+1})$ lies below the point $(p_i,h_i)$ for $i\in\{-m+1,\ldots,m' -2\}$. 
\end{assumption}
We refer to the intervals $[p_i,p_{i+1})$ mentioned in Assumption \ref{assumption:1} above as ``continuity intervals'' of $u$. Without loss of generality, we assume that $p^*\in[p_0,p_1)$.

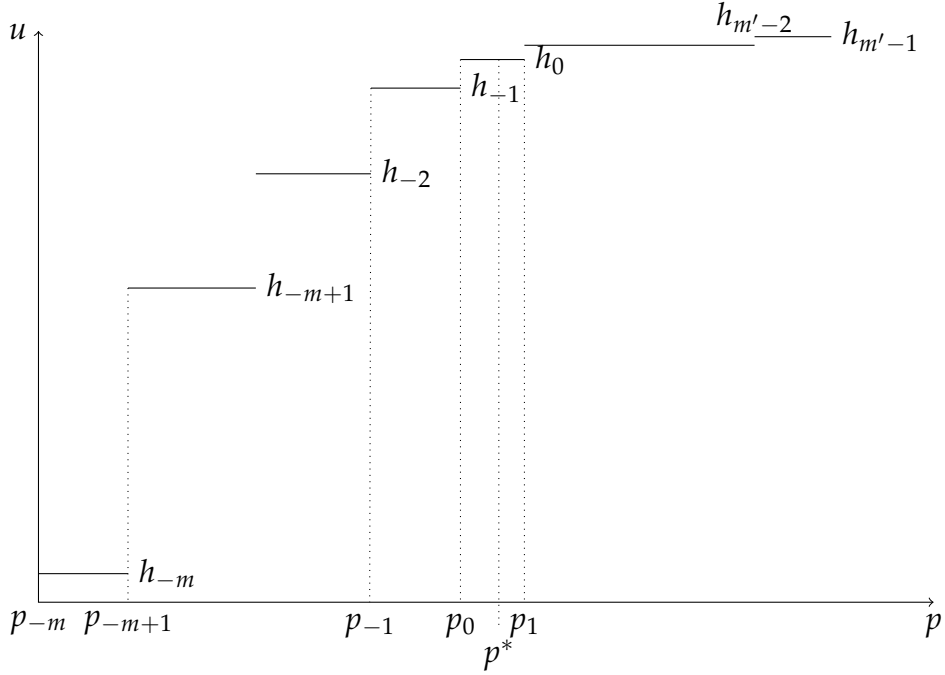
\begin{figure}[h]
\centering
\begin{tikzpicture}[domain=0:1,xscale=17,yscale=3.8]
\draw[<->] (0,2) node[left]{$u$}-- (0,0) -- (.7,0) node[below] {$p$};
\draw[] (0,0.1)--(0.07,0.1) node[right]{$h_{-m}$};
%\draw[red,thin, dashed] (0.06, 0.1)--(0.07,1.1);
\draw[] (0.07,1.1)--(0.17,1.1) node[right]{$h_{-m+1}$};
%\draw[red,thin, dashed] (0.13, 1.1)--(0.135,1.4);
\draw[] (0.17,1.5)--(0.26,1.5) node[right]{$h_{-2}$};
%\draw[red,thin, dashed] (0.25, 1.5)--(0.255,1.8);
\draw[] (0.26,1.8)--(0.33,1.8) node[right]{$h_{-1}$};
%\draw[red,thin, dashed] (0.32, 1.8)--(0.325,1.9);
\draw[] (0.33,1.9)--(0.38,1.9) node[right]{$h_{0}$};
%\draw[red,thin, dashed] (0.37, 1.9)--(0.375,1.95);
\draw[] (0.38,1.95)--(0.56,1.95) node[above]{$h_{m' -2}$};
%\draw[red,thin, dashed] (0.55, 1.95)--(0.555,1.98);
\draw[] (0.56,1.98)--(0.62,1.98) node[right]{$h_{m' -1}$};
\draw[dotted] (0.07,1.1)--(0.07,0);
\node at ( 0.07,-0.07) [] {$p_{-m+1}$};
\node at ( 0,-0.07) [] {$p_{-m}$};
%\draw [->] (0.09,-0.12) -- (0.067,0);
%\node at ( -0.05,-0.12) [] {$p_{-m+1}-\varepsilon$};
%\draw [->] (0,-0.08) -- (0.06,0);
\draw[dotted] (0.26,1.8)--(0.259,0) node[below]{$p_{-1}$};
\draw[dotted] (0.33,1.9)--(0.33,0) node[below]{$p_{0}$};
\draw[dotted] (0.36,1.9)--(0.36,-0.09) node[below]{$p^*$};
\draw[dotted] (0.38,1.9)--(0.38,0) node[below]{$p_{1}$};
%\draw[dotted] (0.562,1.9)--(0.562,0) node[below]{$p_{n-1}$};
%\draw[dotted] (0.62,1.9)--(0.62,0) node[below]{$p_{n}$};
\end{tikzpicture}
\caption{The function $u$.} \label{fig:M1}
\end{figure}

\medskip
\noindent\textbf{Game Objective.} 

The length of a period is denoted by $\Delta > 0$.
The sender's objective is to maximize her discounted payoff, 
calculated with respect to her discount factor $r$, ($r>0$).
 The value at the initial belief $p^0$ is
\begin{equation}\label{eq objective disc}
v_\Delta(p^0) := \max_\sigma  \E_{p,\sigma}\left[(1-e^{-r\Delta})\sum_{n=1}^\infty  e^{-r\Delta n}  u(p^n) \right], 
\end{equation}
where the maximum is over all sender's message strategies $\sigma$.
We will denote the game described above by $G_\Delta(u)$.

\medskip
\noindent\textbf{Continuous-Time Game.} 
We are interested in characterizing the value and the sender's optimal message strategy when the length of a period $\Delta$ is small.
To this end, we study the continuous-time game
denoted $G_{cont}(u)$.

To properly relate the games in discrete time to the game in continuous time,
we assume that in $G_\Delta(u)$, 
the per-period probability of switching from state 1 to 0 is $1-e^{-\lambda_1\Delta}$,
and the per-period probability of switching from state 0 to 1 is $1-e^{-\lambda_0\Delta}$, where both $\lambda_0,\lambda_1>0$.
In the game $G_{cont}(u)$, 
the generator of the Markov chain is 
\[ R=\left(\begin{tabular}{ c c }
 $-\lambda_0$ & $\lambda_0$  \\ 
 $\lambda_1$ & $-\lambda_1$ 
\end{tabular}\right), \] 
and the stationary probability of state 1 is\footnote{For this, as well as all the other results on Markov chains used in this paper, see, e.g., \cite{norris_markov_1998}.
}
\[
p^* = \frac{\lambda_0}{\lambda_0+\lambda_1}.
%\omega_1 + \frac{\lambda_1}{\lambda_0+\lambda_1} \omega_0.
\]

In the game $G_{cont}(u)$,
the state variable is the receiver's belief $p^t$,
and this belief determines the receiver's and sender's instantaneous payoffs.
The set of sender's message strategies can therefore be identified
with the set of c\`adl\`ag processes $(p^t)_{t \geq 0}$ with initial belief $\E[p^0]$.

When no information is revealed, the belief changes as a result of the Markov transition as follows:
\begin{equation}\label{eq p der}
   \frac{\partial p^t}{\partial t}=-\lambda_1 p^t + \lambda_0(1-p^t)
   = (\lambda_0+\lambda_1)(p^*-p^t).
\end{equation}
This implies that for every strategy of the sender, 
the process $(p^t)_{t \geq 0}$ satisfies
\begin{equation}
    \label{equ:belief}
\E[p^{t+h} \mid p^t] = p^* + (p^t - p^*) e^{(-(\lambda_0+\lambda_1)h)}, \ \ \ \forall t,h \geq 0. 
\end{equation}

% A simpler family of sender's message strategies are \textit{Markovian} strategies,
% where the play of the sender at each time instance $t$ depends only on the current belief.

% \color{green} /// Zvika:
% I think this last sentence should be deleted because we anyway explain it below when we describe "our contribution''\color{black}

Denote by $v_{cont}$ the value function of $G_{cont}(u)$:
\[ v_{cont}(p) := \sup_\sigma \E\left[\int_0^\infty r e^{-rt}  u(p^t) \rmd t\right], \]
where $\sigma$ ranges over all message strategies of the sender.
We characterize the value  and the sender's optimal message strategy $\sigma^*$ in $G_{cont}(u)$.
We then prove that $v_{cont} = \lim_{\Delta \to 0} v_\Delta$,
and that $\sigma^*$ is approximately optimal in $G_\Delta(u)$, provided $\Delta$ is sufficiently small.

\medskip
\noindent\textbf{Our Contribution.} 
Our main result is Theorem~\ref{theorem:main_result} below, which provides a characterization of the optimal sender's message strategy in the continuous-time game $G_{cont}(u)$, from which we can calculate the value function $v_{cont}$.
Theorem~\ref{theorem:main_result} also implies that the value function in the continuous-time problem approximates the value function in the discrete-time problem, and that a discrete-time approximation of the optimal sender's message strategy in the continuous-time game is approximately optimal in the discrete-time game $G_{\Delta}(u)$.

We show that the optimal sender's message strategy is Markovian: the play at each time instance $t$ depends only on the receiver's belief at that time.
In addition, the optimal strategy involves only two types of behaviors: either the sender reveals no information, 
or the sender sends one of two messages, 
which split the receiver's belief into two possible beliefs.
This property is a consequence of the fact that there are two states.

The theorem shows that the optimal sender's message strategy is different for receiver's beliefs that lie below $p^*$ (where the sender's payoff is smaller than her payoff at $p^*$)
and for beliefs that lie above $p^*$ (where the sender's payoff is larger than her payoff at $p^*$).
For any receiver's belief that belongs to a continuity interval $[p_{-j},p_{-j+1})$ that lies to the left of the invariant distribution $p^*$ as well as for beliefs that belong to the continuity interval $[p_0,p_1)$ that contains the invariant distribution
(when $p^* \in (p_0,p_1)$),
the sender splits the receiver's belief between the endpoints of the interval that contains it.
For receiver's beliefs that belong to continuity intervals $[p_j,p_{j+1})$ that lie above the interval $[p_0,p_1)$, the sender's optimal behavior is different: for each such continuity interval, there is a cutoff $q_j \in (p_j,p_{j+1})$ such that at beliefs in $[p_j,q_j]$ the sender reveals no information,
while at receiver's beliefs in $[q_j,p_{j+1})$ the senders splits the belief between $q_j$ and $p_{j+1}$.
Moreover, the optimal strategy in the continuous-time problem 
is almost optimal in the discrete-time problem,
provide $\Delta$ is sufficiently small.

\begin{theorem}
Suppose that the indirect payoff function $u$ satisfies Assumption \ref{assumption:1}. 
The game in continuous time $G_{cont}(u)$ admits a value function $v_{cont}$,
and the following Markovian message strategy  $\sigma^*$ of the sender is optimal in the continuous-time game:
\begin{itemize}
\item
If $p^* = p_0$, then at $p^*$ the sender reveals no information.
\item
If $p^* \in (p_0,p_1)$, then at every $p \in [p_0,p_1]$ the sender splits the belief into $p_0$ and $p_1$.
\item
For every $j\in\{0,\ldots,m-1\}$ and every $p\in[p_{-(j+1)},p_{-j})$, at $p$ the sender splits the belief into $p_{-(j+1)}$ and $p_{-j}$.
\item
For every $j \in\left\{1,\ldots,m'-2\right\}$, 
there is $q_j \in (p_j,p_{j+1})$ such that 
\begin{itemize}
    \item at every receiver's belief $p \in [p_j,q_j]$, the sender reveals no information, and
    \item for every receiver's belief $p \in (q_j,p_{j+1}]$, at $p$ the sender splits the receiver's belief $p$ into $q_{j}$ and $p_{j+1}$.
\end{itemize}
\end{itemize}
Moreover, for every $\ep > 0$ there is $\Delta_0 > 0$
such that $\sigma^*$ is $\ep$-optimal when the gap between stages
is $\Delta$, for every $\Delta \in (0,\Delta_0)$. 
\label{theorem:main_result}
\end{theorem}

The intuition for this result is as follows.
Suppose that the receiver's current belief is $p < p^*$,
and the sender would like to have the belief $p$ reach some belief $p' \in (p,p^*]$.
There are two simple ways in which the sender can achieve this goal: 
(i) she can reveal no information, and let the belief slide towards $p'$
because of the Markov transition,
or (ii) she can let the belief move slightly towards $p^*$ because of the Markov transition
and immediately reveal information to the receiver in such a way that the receiver's belief is split between $p$ and $p'$.
It turns out that when discounting is taken into account,
repeated splitting of the receiver's belief achieves a faster convergence to the belief $p'$ than sliding.
When $p < p^*$, the monotonicity of the sender's payoff implies that repeated splitting is superior to sliding, and the concavity of the payoffs implies that it is optimal for the sender to
split the belief between $p$ and the discontinuity point of $u$ to its right.
When $p > p' \geq p^*$, repeatedly splitting the belief between $p$ and $p'$ still converges to $p'$ faster than sliding, but now the monotonicity of $u$ does not imply that sliding is better than splitting.
On the one hand, to reach from $p_{j+1}$ to $p_j$, sliding yields the sender the payoff $h_j$ for as long as possible.
On the other hand, repeated splitting allows the sender to obtain the higher payoff $h_{j+1} = u(p_{j+1})$. 
For beliefs $p$ on the continuity interval $[p_j,p_{j+1}]$ that are close to $p_j$, delaying the arrival to $p_j$ by revealing no information turns out to be optimal. The situation is reversed for beliefs in this continuity interval that are close to $p_{j+1}$. For such beliefs, 
splitting between $p_j$ and $p_{j+1}$ generates the belief $p_{j+1}$ with high probability, which implies that splitting is better than sliding.

\medskip
\noindent\textbf{Comparison with Existing Literature.} 
\cite{ely_beeps_2017} and \cite{Renault_2017} studied the model 
with a single discontinuity point (where concavity has no bite).
Their model corresponds to our model:
if $p^*$ is below the discontinuity point, then $m=0$ and $m'=2$;
if $p^*$ is at least the discontinuity point and less than 1, then $m=m'=1$;
and if $p^*=1$, then $m=2$ and $m'=0$.
These authors
proved that the myopic strategy is optimal.
In our setup, the myopic message strategy of the sender uses binary splits below $p_0$ and can be either sliding or splitting above $p_0$.
Theorem~\ref{theorem:main_result} shows that when $m' > 1$,
the sender's optimal message strategy is myopic at $p^*$,
and at all continuity intervals to the left of $p^*$,
but is not myopic at the continuity intervals that lie below $p^*$.
\cite{ashkenazi-golan_solving_2020} studied the model when $u$ is continuous (rather than piecewise constant and concave),
and provided an algorithm for calculating the value function and the sender's optimal message strategy in the continuous-time game.
%Theorem~\ref{theorem:main_result} extends this work to discontinuous $u$.
\cite{cardaliaguet_markov_2016} studied the model with finitely many states and continuous payoff function,
characterized the value function of the continuous-time game as the viscosity solution of a certain equation,
and proved that the value of the discrete-time game converges to the value of the continuous-time game as the inter-stage duration goes to 0.
Theorem~\ref{theorem:main_result} extends the 
results of \cite{ashkenazi-golan_solving_2020} and the
approximation result of \cite{cardaliaguet_markov_2016} to discontinuous $u$.

\subsection{Sketch of the Proof}

In this subsection we highlight the main ideas behind the proof of Theorem~\ref{theorem:main_result}. 
The detailed proof appears in Section~\ref{section:analysis},
and technical aspects are relegated to the appendix.

The value function $v_{cont}$ was studied and characterized by \cite{cardaliaguet_markov_2016}, for the case in which the payoff function $u$ is continuous. 
When the payoff function is not continuous, 
as is the case here,
their characterization of $v_{cont}$ is not valid,
and the existence of a sender optimal strategy is not guaranteed. 

 We use the result of \cite{cardaliaguet_markov_2016} to prove the existence of the value for indirect payoff functions $u$ that satisfy Assumption~\ref{assumption:1}, and to characterize the sender's optimal message strategy. 
This is done by bounding the discontinuous function $u$ by continuous functions $\overline{u}_\delta$, which approach $u$ from above as $\delta$ decreases to zero.   
We show that  as $\delta$ decreases to zero the respective values, $\overline{v}_\delta$ converge to the value that is obtained by using the strategy $\sigma^*$ when the payoff function is the discontinuous function $u$. From the monotonicity of the value, this is also the value for the discontinuous payoff function $u$. 

The main result, including the construction of $\sigma^*$, is presented in
Section~\ref{section:analysis}. In Section~\ref{sec: u approx}, we introduce continuous payoff functions $\overline{u}_\delta$, that are higher of equal to $u$ and approximate $u$ as $\delta$ decreases to 0. In Section~\ref{section optimal inn aprox}, we characterize the optimal strategies for games with payoff functions $\overline{u}_\delta$. Section~\ref{sec optimal st gcont} returns to the value function $u$ and specifies the optimal strategy in continuous time. In Section~\ref{section:approx}, we show that the sender's optimal message strategy for the continuous-time game has a close strategy which is approximately optimal for the discrete-time game.  

\section{Analysis}
\label{section:analysis}

In this section we present the detailed proof of Theorem~\ref{theorem:main_result}. 
As mentioned before, the results of \cite{cardaliaguet_markov_2016} hold when the indirect payoff function $u$ is continuous. 
When the payoff function is not continuous, 
as in the case here,
their characterization of $v_{cont}$ is not valid,
and the existence of an optimal sender's strategy is not guaranteed. 
We therefore approximate $u$ from above by continuous functions.

\subsection{The Approximating Continuous Functions}\label{sec: u approx}

For every $\delta > 0$ 
such that $\delta < p_{j+1}-p_j$ for every $j=-m,-m+1,\dots,m'-1$,
define a payoff functions $\overline{u}_\delta$ as follows (see Fig.~\ref{fig fuctions}):

\begin{equation*}
    \overline{u}_\delta = \begin{cases}
               h_j,             & p\in[p_j,p_{j+1}-\delta],~ j\in\{-m,\ldots,m'-2\},\\ \left(\frac{h_{j+1}-h_j}{\delta}\right)\cdot (p-p_{j+1}) + h_{j+1},    & p\in[p_{j+1}-\delta, p_{j+1}],~ j\in\{-m,\ldots,m'-2\},  \\
          h_{m'-1}, & p\in[p_{m'-1},p_{m'}].     
               \end{cases}
\end{equation*}

Since $u$ has a concave upper envelope, so does the function $\overline{u}_\delta$. 
The sequence $(\overline u_\delta)_{\delta > 0}$ is nonincreasing (as $\delta$ goes to 0) and converges pointwise to $u$.
Denote by $\overline v_\delta$ the value function of the game $G_{cont}(\overline u_\delta)$.
Since the sequence $(\overline u_\delta)_{\delta > 0}$ is nonincreasing,
the sequence $(\overline v_\delta)_{\delta > 0}$ of value functions is nonincreasing as well.
Denote the limit value function by
\[ \overline v_0(p) := \lim_{\delta \to 0} \overline v_\delta(p), \ \ \ \forall p \in [0,1]. \]
Since $\overline u_\delta \geq u$, we have $\overline v_0 \geq v$,
that is,
\begin{equation}
\label{equ:val:1}
\overline v_0(p) \geq v_{cont}(p), \ \ \ \forall p \in [0,1]. 
\end{equation}
We will prove that in fact Eq.~\eqref{equ:val:1} holds with equality:
the limit of the value functions of the approximating games is the value function of the original problem in continuous time.

\begin{figure}[h]
\centering

\begin{tikzpicture}
    \draw[](2,0.1) -- (2,-0.1)node[below, scale=0.65]{ $p_j$};
    \draw[](1.1,0.1) -- (1.1,-0.1)node[below, scale = 0.65]{$p_j-\delta$};
    %\draw[](2.9,0.1) -- (2.9,-0.1)node[below, scale = 0.65]{$p_j+\delta$};
    \draw[](5.1,0.1) -- (5.1,-0.1)node[below, scale = 0.65]{$p_{j+1}-\delta$};
    \draw[](6,0.1) -- (6,-0.1)node[below, scale=0.65]{$p_{j+1}$};
    % \draw[](6.9,0.1) -- (6.9,-0.1)node[below, scale = 0.65]{$p_{j+1}+\delta$};
      \draw[->] (-0.5,0) -- (8,0) node[right] {$p$};
      \draw[->] (0,-0.5) -- (0,4) node[above] {$u$};
      \draw [blue, dashed] (0.5,0.5) -- (2,0.5);
       \draw [blue, dashed] (2,2) -- (6,2);
       \draw [blue, dashed] (6,3) -- (7.5,3);

      \draw [red, densely  dotted] (0.5,0.51) -- (1.1,0.51);
      \draw[red, densely dotted]
      (1.1,0.51) --
      (2,2);
      \draw[red, densely dotted]
      (2,2) --
      (5.1,2);
       \draw[red, densely dotted]
      (5.1,2) --
      (6,3.01);
      \draw[red, densely dotted]
      (6,3.01) --
      (7.5, 3.01);

      %\draw [blue, dashed] (0.5,0.48) -- (2,0.48);
     % \draw[blue, dashed] (2,0.48) --(2.9,1.98);
     % \draw[blue, dashed](2,1.98) --(6,1.98);
      % \draw[blue, dashed] (6,1.98) --(6.9,2.98);
      %\draw[blue, dashed](6.9,2.98) --(7.5, 2.98);

      \draw [loosely dotted] (2,0) -- (2,2);
       \draw [loosely dotted] (6,0) -- (6,3);

       \node[scale=0.7] at (7, 4.3)   {$\overline{u}_\delta$};
       \node[scale=0.7] at (7, 4)   {$u$};
       %\node[scale=0.7] at (7, 3.7)   {$\underline{u}_\delta$};
       \draw [blue, dashed] (7.3,4) -- (7.9,4);
       %\draw [blue, dashed] (7.3,3.7) -- (7.9,3.7);
       \draw[red, densely dotted]
      (7.3,4.3) --
      (7.9,4.3);
\end{tikzpicture}
\caption{The continuous payoff functions $\overline{u}_\delta$, approximating $u$ from above.} \label{fig fuctions}
\end{figure}
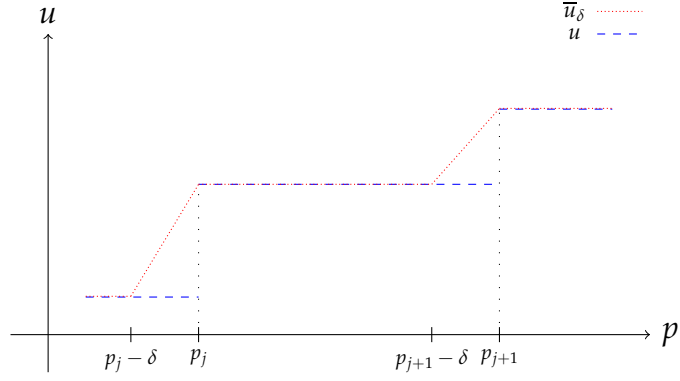

\subsection{Characterizing the Optimal Strategy in $G_{cont}(\overline u_\delta)$}\label{section optimal inn aprox}

The heart of the proof is the characterization of the sender's optimal message strategy $\overline\sigma_\delta$ in the game
$G_{cont}(\overline u_\delta)$, for $\delta > 0$ sufficiently small, which is displayed in Figure~\ref{fig strategy}.
The strategy $\overline\sigma_\delta$ is Markovian;
for each $j\in\{1,\dots,m\}$ there is a real number $\overline {q}_{-j}(\delta) \in (p_{-j}-\delta,p_{-j})$
such that in the interval $[p_{-j-1},\overline q_{-j}(\delta)]$ the sender splits the receiver's belief between the 
endpoints of the interval, 
and in the interval $[\overline q_{-j}(\delta),p_{-j}]$ the sender reveals no information.
Similarly, for each $j=1,\dots,m'-1$ there is a real number 
$\overline q_{j}(\delta) \in (p_{j},p_{j+1})$
such that in the interval $[p_{j},\overline q_{j}(\delta)]$ the sender reveals no information,
and in the interval $[\overline q_{j}(\delta),p_{j+1}]$
the sender splits the receiver's belief between the endpoints of the interval.

\begin{figure}[h]
\centering

\begin{tikzpicture}[scale=1.5]
\draw [latex-latex](0.1,0.5) -- (1.4,0.5);
\draw [latex-latex](1.5,0.5) -- (2,0.5);
\draw [latex-latex](5,0.5) -- (5.5,0.5);
\draw [latex-latex](5.6,0.5) -- (7,0.5);
\draw [latex-latex](3,0.5) -- (4,0.5);

    \draw[](3,0.1) -- (3,-0.1)node[below, scale=0.65]{ $p_{0}$};
    \draw[](4,0.1) -- (4,-0.1)node[below, scale=0.65]{ $p_{1}$};
    \draw[](0,0.1) -- (0,-0.1)node[below, scale=0.65]{ $p_{-j-1}$};
    \draw[](2,0.1) -- (2,-0.1)node[below, scale=0.65]{ $p_{-j}$};
    \draw[](1,0.1) -- (1,-0.1)node[below, scale = 0.65]{$p_{-j}-\delta$};
    \draw[dotted](1.5,0.1) -- (1.5,-0.5)node[below, scale = 0.65]{$\overline q_{-j}(\delta)$};
    %\draw[](2.9,0.1) -- (2.9,-0.1)node[below, scale = 0.65]{$p_j+\delta$};
    \draw[dotted](5.5,0.1) -- (5.5,-0.5)node[below, scale = 0.65]{$\overline q_{j}(\delta)$};
    \draw[](5,0.1) -- (5,-0.1)node[below, scale = 0.65]{$p_{j}$};
%    \draw[](6,0.1) -- (6,-0.1)node[below, scale = 0.65]{$p_{j+1}-\delta$};
    \draw[](7,0.1) -- (7,-0.1)node[below, scale=0.65]{$p_{j+1}$};
    % \draw[](6.9,0.1) -- (6.9,-0.1)node[below, scale = 0.65]{$p_{j+1}+\delta$};
      \draw[->] (-0.5,0) -- (8,0) node[right] {$p$};

       \node[scale=0.7] at (.75, .8)   {split};
       \node[scale=0.7] at (1.75, .8)   {slide};
       \node[scale=0.7] at (3.5, .8)   {split};
       \node[scale=0.7] at (5.25, .8)   {slide};
       \node[scale=0.7] at (6.25, .8)   {split};

      %\draw [blue, dashed] (0.5,0.48) -- (2,0.48);
     % \draw[blue, dashed] (2,0.48) --(2.9,1.98);
     % \draw[blue, dashed](2,1.98) --(6,1.98);
      % \draw[blue, dashed] (6,1.98) --(6.9,2.98);
      %\draw[blue, dashed](6.9,2.98) --(7.5, 2.98);

\end{tikzpicture}
\caption{The characterization of the optimal strategy in $G_{cont}(\overline u_\delta)$.} \label{fig strategy}
\end{figure}
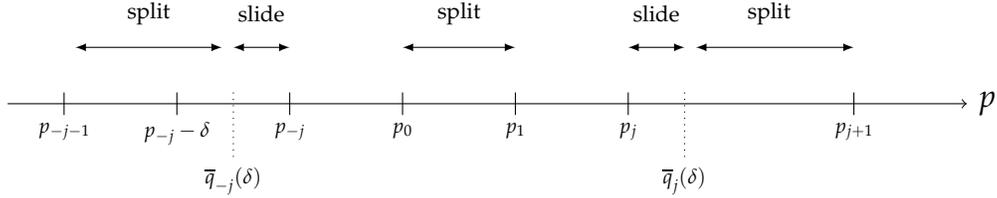

The formal statement follows.

\begin{lemma}
\label{lemma:structure}
Let $\delta > 0$ be sufficiently small.
For every $j\in\{0,1,\ldots,m-1\}$ there exists $\overline q_{-j}(\delta) \in (p_{-j}-\delta,p_{-j})$,
and for every $j\in\{1,2,\ldots,m'-1\}$ there exists $\overline q_{j}(\delta) \in (p_{j},p_{j+1})$,
such that the sender's optimal message strategy in $G_{cont}(\overline u_\delta)$,
denoted  $\overline{\sigma}^*_{\delta}$, is as follows:
\begin{itemize}
    \item For $p\in[p_0,p_1]$, split the belief between $p_0$ and $p_1$.
    \item For  $p\in[p_{-j-1},\overline{q}_{-j}(\delta)]$, split the belief between $p_{-j-1}$ and $\overline{q}_{-j}(\delta)$, 
    for every $j\in\{1,\ldots,m-1\}$ and $j=0$ if $p^*>p_0$.
    \item For $p\in[\overline{q}_{-j}(\delta), p_{-j}]$, reveal no information, 
    for every $j\in\{1,\ldots,m-1\}$ and $j=0$ if $p^*>p_0$.
 \item If $p^*=p_0$, then for $p\in[p_{-1},p_0)$ split the belief between $p_{-1}$ and $p_0$, 
 for $p\in (p_0,p_1]$ split the belief between $p_0$ and $p_1$,
 and for $p=p^*$ reveal no information.
    \item For $p\in[p_j,\overline{q}_j(\delta)]$ reveal no information, 
    for every $j\in\{1,\ldots,m'-1\}$.
    \item For $p\in[\overline{q}_j(\delta),p_{j+1}]$, split the belief between $\overline{q}_j(\delta)$ and $p_{j+1}$, 
    for every $j\in\{1,\ldots,m'-1\}$.
\end{itemize}
\end{lemma}

The proof of Lemma~\ref{lemma:structure} requires a careful analysis of the characterization 
of the value function due to \cite{cardaliaguet_markov_2016} and \cite{gensbittel_continuous-time_2016},
and relies on the special structure of the payoff function
that is implied by Assumption~\ref{assumption:1}. 
The proof is relegated to the Appendix.
\subsection{The Optimal Strategy in $G_{cont}(u)$}\label{sec optimal st gcont}

As discussed in Lemma~\ref{lemma:structure}, 
the strategy $\overline\sigma_\delta$ is determined by the cut-offs
$(\overline q_{-j}(\delta))_{j=0}^{m-1}$ and $(\overline q_j(\delta))_{j=1}^{m'-1}$.
By taking a subsequence,
we can assume w.l.o.g.~that the limits
\[ q_j := \lim_{\delta \to 0} \overline q_j(\delta), \ \ \ j \in \{-m+1,\dots,m'-1\}, \]
exist. 
Moreover, for $j \in \{-m+1,\cdots,0\}$
we have $p_j - \delta < \overline q_j(\delta) < p_j$, hence
$\lim_{\delta \to 0} \overline q_j(\delta)=p_j $.
Let $\sigma^*$ be the strategy that is defined by these limits, see Figure~\ref{fig strategy2}.

\begin{itemize}
    \item For $p\in[p_0,p_1]$ split the belief between $p_0$ and $p_1$.
    \item For  $p\in[p_{-j-1},p_{-j}]$, split the belief between $p_{-j-1}$ and $p_{-j}$, 
    for every $j
    \in\{0,\ldots,m-1\}$.
 \item If $p^*=p_0$, then for $p\in[p_{-1},p_0)$ split the belief between $p_{-1}$ and $p_0$, 
 for $p\in (p_0,p_1]$ split the belief between $p_0$ and $p_1$.
    \item For $p\in[p_j,q_j]$ reveal no information, 
    for every $j\in\{1,\ldots,m'-1\}$.
    \item For $p\in[q_j,p_{j+1}]$, split between $q_j$ and $p_{j+1}$, 
    for every $j\in\{1,\ldots,m'-1\}$.
\end{itemize}

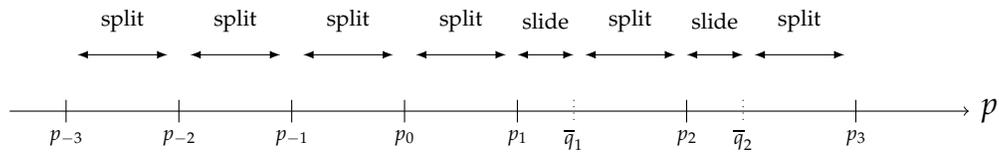
\begin{figure}[h]
\centering

\begin{tikzpicture}[scale=1.5]
\draw [latex-latex](0.1,0.5) -- (0.9,0.5);
\draw [latex-latex](1.1,0.5) -- (1.9,0.5);
\draw [latex-latex](2.1,0.5) -- (2.9,0.5);
\draw [latex-latex](3.1,0.5) -- (3.9,0.5);
\draw [latex-latex](4,0.5) -- (4.5,0.5);
\draw [latex-latex](4.6,0.5) -- (5.4,0.5);
\draw [latex-latex](5.5,0.5) -- (6,0.5);
\draw [latex-latex](6.1,0.5) -- (6.9,0.5);

    \draw[](0,0.1) -- (0,-0.1)node[below, scale=0.65]{ $p_{-3}$};
    \draw[](1,0.1) -- (1,-0.1)node[below, scale=0.65]{ $p_{-2}$};
    \draw[](2,0.1) -- (2,-0.1)node[below, scale=0.65]{ $p_{-1}$};
    \draw[](3,0.1) -- (3,-0.1)node[below, scale=0.65]{ $p_{0}$};
    \draw[](4,0.1) -- (4,-0.1)node[below, scale=0.65]{ $p_{1}$};
    \draw[dotted](4.5,0.1) -- (4.5,-0.1)node[below, scale = 0.65]{$\overline q_{1}$};
    \draw[](5.5,0.1) -- (5.5,-0.1)node[below, scale = 0.65]{$p_2$};
    \draw[dotted](6,0.1) -- (6,-0.1)node[below, scale = 0.65]{$\overline q_{2}$};
    \draw[](7,0.1) -- (7,-0.1)node[below, scale = 0.65]{$p_3$};
      \draw[->] (-0.5,0) -- (8,0) node[right] {$p$};

       \node[scale=0.7] at (.5, .8)   {split};
       \node[scale=0.7] at (1.5, .8)   {split};
       \node[scale=0.7] at (2.5, .8)   {split};
       \node[scale=0.7] at (3.5, .8)   {split};
       \node[scale=0.7] at (4.25, .8)   {slide};
       \node[scale=0.7] at (5, .8)   {split};
       \node[scale=0.7] at (5.75, .8)   {slide};
       \node[scale=0.7] at (6.5, .8)   {split};

      %\draw [blue, dashed] (0.5,0.48) -- (2,0.48);
     % \draw[blue, dashed] (2,0.48) --(2.9,1.98);
     % \draw[blue, dashed](2,1.98) --(6,1.98);
      % \draw[blue, dashed] (6,1.98) --(6.9,2.98);
      %\draw[blue, dashed](6.9,2.98) --(7.5, 2.98);

\end{tikzpicture}
\caption{The strategy $\sigma^*$ in $G_{cont}( u)$.} \label{fig strategy2}
\end{figure}

%Since the sequence $(\overline u_\delta)_{\delta > 0}$ is decreasing with $\delta$,
%the sequence of corresponding values $(\overline v_\delta)_{\delta > 0}$ decreases as well.
%denote
%\[ \overline v_0(p) := \lim_{\delta \to 0} \overline v_\delta(p), \ \ \ \forall p \in [0,1]. \]
As the following result states, the payoff under $\sigma^*$ in $G_{cont}(u)$ is $\overline v_0$.
This holds because the behavior of the sender under $\overline \sigma^*_\delta$ converges to her behavior under $\sigma^*$.

For every message strategy $\sigma$ of the sender and every $p \in [0,1]$, denote the putative value  under $\sigma$ at $p$ in $G_{cont}(u)$ by
\begin{equation}\label{eq objective cont}
w(p,\sigma) := \E_\sigma\left[\int_{t=0}^\infty  e^{-rt}  u(p^t) dt \mid p^0 = p\right], 
\end{equation}

\begin{lemma}\label{lemma:continuity}
For every $p \in [0,1]$ we have $\overline v_0(p) = w(p,\sigma^*)$.
\end{lemma}

Since there is a sender's message strategy that guarantees the payoff $\overline v_0$,
we have 
\[ v_{cont}(p) \geq \overline v_0(p), \ \ \ \forall p \in [0,1]. \]
Together with Eq.~\eqref{equ:val:1} this implies that 

\begin{equation}\label{eq equal values}
v_{cont} = \overline v_0,   
\end{equation} 
and that $\sigma^*$ is an optimal strategy in $G_{cont}(u)$.

The proof is inductive: 
We show that $\overline v_0 = w(\cdot,\sigma^*)$ on $[p_0,p_1]$,
and continue to show the same on continuity intervals $[p_{-j-1},p_{-j}]$ and $[p_j,p_{j+1}]$ with larger $j$'s.
The proof appears in Section~\ref{sec proof lemma conti}.

\subsection{The Value Function in Continuous Time as an Approximation of the Value Function in Discrete Time}
\label{section:approx}

So far we characterized the value function
and the sender's optimal message strategy in the continuous-time game.
Here we complete the proof of Theorem~\ref{theorem:main_result}, 
by showing 
that $\sigma^*$, when interpreted as a strategy in  the discrete-time game $G_\Delta(u)$, is approximately optimal, 
provided $\Delta$ is sufficiently small. 

\cite{cardaliaguet_markov_2016} proved that 
if the instantaneous payoff function $u$ is Lipschitz for the $L^1$-norm,
then the sender's optimal message strategy in the continuous-time game is approximately optimal in $G_\Delta(u)$,
provided $\Delta$ is sufficiently small. 
In our model $u$ is not continuous, hence we cannot apply the results of \cite{cardaliaguet_markov_2016}.
Our proof is divided into two steps. 
Lemma~\ref{lemma discrete ap} states that the payoff under the strategy $\sigma^*$ in $G_\Delta(u)$ 
approaches $v_{cont}$, the payoff under $\sigma_*$ in $G_{cont}(u)$ as $\Delta$ goes to 0. 
Lemma~\ref{lemma diminish dif} then implies that $v_{cont}$ is close to the value of the discrete-time game.

The strategy $\sigma^*$ belongs to in the continuous-time game
$G_{cont}(u)$:
For each receiver's belief $p$ it indicates whether the sender reveals no information,
or whether she splits the receiver's belief between two beliefs.
However, it can be viewed also as a strategy in the discrete-time game 
$G_\Delta(u)$:
At every stage $n$, as a function of the current receiver's belief $p^n$,
it reveals no information if $\sigma^*$ reveals no information at $p^n$,
and otherwise it splits the belief as $\sigma^*$ does.
To avoid cumbersome notations, we denote the strategy induced in the discrete-time game by $\sigma^*$ as well.

Denote the putative value obtained by the sender under the message strategy $\sigma^*$ 
in the discrete-time game $G_\Delta(u)$ by $w_\Delta(p,\sigma^*)$. 

\begin{lemma}\label{lemma discrete ap}
For every $p \in [0,1]$, $\lim_{\Delta\to 0}w_\Delta(p,\sigma^*)=v_{cont}(p)$. 
\end{lemma}
The proof of Lemma~\ref{lemma discrete ap} appears in the appendix, in Section~\ref{sec proof lemma lim}.

Lemma~\ref{lemma discrete ap} implies that 
$\lim_{\Delta \to 0} v_\Delta  \geq v_{cont}$.
The next lemma states that 
$\lim_{\Delta \to 0} v_\Delta \leq v_{cont}$,
thereby completing the proof of Theorem~\ref{theorem:main_result}.

\begin{lemma}\label{lemma diminish dif}
For every $p\in[0,1]$ and every sender's message strategy $\sigma$, $\lim_{\Delta\to 0} w_\Delta(p,\sigma)\leq v_{cont}(p)$.
\end{lemma}

The proof of Lemma~\ref{lemma diminish dif} appears in the appendix, in Section~\ref{sec proof lemma dim}.

\section{Discussion}

Assumption~\ref{assumption:1} requires that the payoff function is monotone, piecewise constant, and has concave envelope.
How does the characterization of the sender's optimal message strategy change when this assumption is weakened?

The case that $u$ is continuous (rather than piecewise constant) and concave 
falls under the model studied by \cite{cardaliaguet_markov_2016},
who showed that in this case
the sender's optimal message strategy is to never reveal information.
This can be viewed as a limit case of our model,
when the set of discontinuity points of $u$ becomes dense in $[0,1]$,
in which case the splits of beliefs become narrow
(that is, the sender splits the receiver's belief into nearby beliefs),
and at the limit no splitting is done.

When $u$ is monotone and piecewise constant but not concave,
it is no longer true that for every $p$ at which it is optimal to split the receiver's belief, the optimal split is to the endpoints of the continuity interval that contains $p$.
Indeed, if this interval lies below $p^*$ and does not intersect the concave envelope of $u$,
it will be "skipped" and the receiver's belief will never lie in this interval (after the initial split).

%\label{section:example2}
%
%Consider the following indirect payoff function $u$, with $p^*=\frac{1}{3}$:
%
%\[u(p) = \left\{
%     \begin{array}{cl} 0  & p\in[0,\frac{1}{4}),\\
%     0.1 & p\in[\frac{1}{4},\frac{1}{2}), \\ 
%     0.2 &p\in[\frac{1}{2},\frac{3}{4},\\
%     1 &p\in[\frac{3}{4},1].
%     \end{array}
%   \right.\]
%
%\begin{figure}[h]  
%\centering 
%
%\begin{tikzpicture} 
%   \draw[](1.33,0.1) -- (1.33,-0.3)node[below]{$p^*$}; 
%    \draw[](1,-0.1) -- (1,0.1)node[below]{$\tfrac{1}{4}$}; 
%    \draw[](2,-0.1) -- (2,0.1)node[below]{$\tfrac{2}{4}$}; 
%    \draw[](3,-0.1) -- (3,0.1)node[below]{$\tfrac{3}{4}$}; 
%      \draw[->] (-0.5,0) -- (4,0) node[right] {$p$}; 
%      \draw[->] (0,-0.5) -- (0,4) node[above] {$u$};  
%      \draw [blue, ultra thick] (0,0) -- (1,0); 
%       \draw [blue, ultra thick] (1,0.3) -- (2,0.3);
%       \draw [blue, ultra thick] (2,0.6) -- (3,0.6);
%       \draw [blue, ultra thick] (3,3) -- (4,3); 
%       \draw[](-0.1,3) -- (0.1,3)node[left]{$1$};
%       \draw [dashed] (1,0) -- (1,0.3);
%       \draw [dashed] (2,0) -- (2,0.6);
%       \draw [dashed] (3,0) -- (3,3);
%\end{tikzpicture} 
%\caption{The function $u$ in Section~\ref{section:example2}.} \label{fig example 1}  
%\end{figure}  
%
%One can verify that $(p^*,(\cav u)(p^*))$ is on the line connecting $(0,0)$ and $(\frac{3}{4},1)$. This means that the optimal %strategy for all $p\in[0,\frac{3}{4}]$ is to split between $0$ and $\frac{3}{4}$. That is, the optimal strategy never reaches %the two central parts of $u$, unlike the case where the envelope of the indirect payoff function $u$ is concave.

We discuss next the case when $u$ is piecewise constant with a concave envelope
but not monotone.
Consider for example the choice (see Figure~\ref{fig example 2}):

\[u(p) = \left\{
     \begin{array}{cl} 0,  & p\in[0,\frac{1}{3}),\\
     1, & p\in[\frac{1}{3},\frac{2}{3}), \\ 
     0, &p\in[\frac{2}{3},1].
     \end{array}
   \right.\]

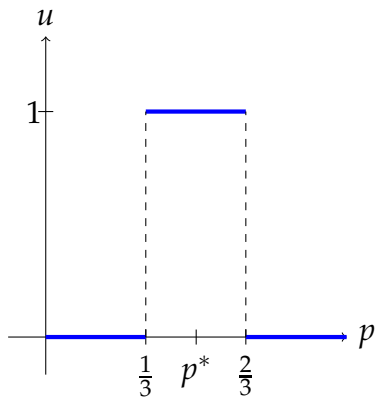
\begin{figure}[h]  
\centering 

\begin{tikzpicture} 
    \draw[](2,0.1) -- (2,-0.1)node[below]{$p^*$}; 
    \draw[](1.33,0.1) -- (1.33,-0.1)node[below]{$\tfrac{1}{3}$}; 
    \draw[](2.66,0.1) -- (2.66,-0.1)node[below]{$\tfrac{2}{3}$}; 
      \draw[->] (-0.5,0) -- (4,0) node[right] {$p$}; 
      \draw[->] (0,-0.5) -- (0,4) node[above] {$u$};  
      \draw [blue, ultra thick] (0,0) -- (1.33,0); 
       \draw [blue, ultra thick] (1.33,3) -- (2.66,3);
       \draw [blue, ultra thick] (2.66,0) -- (4,0); 
       \draw[](-0.1,3) -- (0.1,3)node[left]{$1$};
       \draw [dashed] (1.33,0) -- (1.33,3);
       \draw [dashed] (2.66,0) -- (2.66,3);
\end{tikzpicture} 
\caption{A nonmonotone function $u$.} \label{fig example 2}  
\end{figure}  

Over the belief interval $[0,\frac{1}{2}]$, the function fits the model analyzed in the paper, and therefore the optimal strategy for $p \in [0,\frac{1}{3}]$ is to split the belief between $0$ and $\frac{1}{3}$. For symmetric reason, the optimal strategy for $p \in [\frac{2}{3},1]$ is to split the belief between $\frac{2}{3}$ and $1$. This in unlike the case when $u$ is increasing, where to the right of $p^*$ the  optimal strategy involves sliding as well as splitting.  

In this example, the maximum of $u$ is attained at $p^*$.
When the maximum of $u$ is attained, say, at $p_j$ for $j \geq 1$,
the sender's optimal message strategy on $[0,p_j]$ coincides with $\sigma^*$,
yet the sender's optimal message strategy on $[p_j,1]$ may be more intricate than $\sigma^*$.

For a general piecewise constant $u$,
some continuity intervals will be skipped altogether,
in some the receiver's belief will be split between the two endpoints of the interval
(as happens in our model for continuity intervals below $p^*$),
and some will be divided into two (as happens in our model for continuity intervals above $p^*$): 
in one part no information will be revealed, 
and in the other the belief will be split between 
the interval's cutoff point and some discontinuity point of $u$,
which may or may not be the endpoint of the continuity interval.

A natural question is whether the sender has a \textit{uniformly $\ep$-optimal} message strategy;
that is, 
a strategy that is $\ep$-optimal for every discount rate $r$ sufficiently close to 0.
In this case, only the far future matters.
By Theorem~\ref{theorem:main_result}, if $p^* \in (p_0,p_1)$, 
under $\sigma^*$ the receiver's belief after the sender sends her message converges to $\{p_0,p_1\}$ with probability 1,
and any strategy under which the belief after the sender sends her message converges to $\{p_0,p_1\}$ with probability 1 is uniformly $\ep$-optimal.
If $p^* = p_0$, the same holds for any message strategy under which the receiver's belief after the sender sends her message converges to $\{p_0\}$.

Another interesting question concerns a variation of the model,
where the receiver obtains information about the state at random times, independently of the sender's choices.
We conjecture that the sender's optimal message strategy will have the same structure as $\sigma^*$,
yet the cutoffs $(q_j)_{j=1}^{m'-1}$ will be higher than the ones we identified,
to compensate for the lower significance of the instantaneous payoff.

\section{Conclusion}

Our result is part of the growing literature devoted to the study of optimal strategies in persuasion games.
There are two aspects that single out our work.
First, as in \cite{ely_beeps_2017} and \cite{Renault_2017}, the payoff function in our model is not continuous, but piecewise constant. \cite{ely_beeps_2017} and \cite{Renault_2017} were interested in situations where the belief space is divided into two convex regions
and the payoff in each region is constant,
and asked whether a specific strategy, 
namely, the myopic strategy, is optimal.
In contrast, we allow for more than two continuity regions.
It turns out that in the interval $[0,p_1]$ (or $[0,p_0]$, if $p^*=p_0$)
the sender's optimal message strategy is myopic,
while on the interval $[p_1,1]$ it is not.
Our work highlights the interplay between the Markov transition and the monotonicity of payoffs:
When the Markov transition leads to beliefs with higher (resp.~lower) payoff,
the myopic strategy is optimal (resp.~not optimal). 
Second, our study combines tools provided by the literature on continuous-time games,
like the approach taken by \cite{gensbittel_2020}
with geometric intuitions.

We studied the model with two states of nature. 
A natural question is whether similar analysis can be carried out in the presence of more than two states of nature.
Unfortunately, the answer is negative.
\cite{Renault_2017} presented a 
discrete-time
example with three states of nature where the function $u$ is piecewise constant and attains two values, and the optimal strategy is quite involved.
A similar phenomenon occurs in continuous-time games, as studied by \cite{gensbittel_2020}.

\newpage

\section{Appendix}

The most challanging part of the proof of Theorem~\ref{theorem:main_result} is Lemma~\ref{lemma:structure}. 
The proof of the lemma is organized as follows: In Section~\ref{sec two use} we present two useful 
Markovian strategies: 
one reveals no information in a certain range $[p',p'']$ of beliefs,
and the other splits the receiver's beliefs between $p'$ and $p''$ whenever the current belief is in $[p',p'']$. 
We study the payoff under these strategies in the continuous-time game.
In Section~\ref{sec monoton} we present results about the monotonicity of the value function. 
Section~\ref{sec:str cont} presents useful results from the literature and provides the sender's optimal message strategy for beliefs in $[p_0,p_1]$. 
Section~\ref{sec g} introduces a 
certain function, denoted $\overline g_\delta$, 
and presents its relation with the derivative of the value function.  
Section~\ref{sec der} characterizes the derivative of the value function using $\overline g_\delta$. Section~\ref{sec str and der} connects the strategies presented in Section~\ref{sec two use} to the derivative found in Section~\ref{sec der}.  
Section~\ref{sec lemma cont proof} concludes the proof.
Section~\ref{sec rem proofs} contains 
the proof of intermediate results used in earlier sections, 
and Section~\ref{sec proofs three} presents the proofs of Lemmas~\ref{lemma:continuity}, \ref{lemma discrete ap}, and~\ref{lemma diminish dif}.

\subsection{Two useful strategies}\label{sec two use}

In this section we present two message strategies of the sender that induce different ways for the belief of the receiver to get from one value $p'$ to another value $p''$,
where either $p'< p'' \leq p^*$ or $p'>p''\geq p^*$. 
%that is, the belief moves toward $p^*$. 
%One strategy reveals no information and thus lets the belief slide from $p'$ to $p''$. The other strategy splits the belief repeatedly between $p'$ and $p''$, and thus never gets to beliefs in $(p',p'')$. 
We then compute the expected discounted time it takes for the belief to get from $p'$ to $p''$
under each of the two strategies. 
The purpose of this computation is threefold. 
First, this will allow us to study properties of the optimal strategy. 
Second, the analysis supports
an intuitive explanation for the sender's optimal message strategy
provided in the introduction.
Third, this analysis is general and does not depend on the function $u(p)$, so it might be of independent interest.

For every 
two distinct beliefs $p',p'' \in [0,1]$,
let $\sigma^{split}_{p',p''}$ be a sender's message strategy that splits the receiver's belief between $p'$ and $p''$ (whenever the receiver's belief is in $[p',p'']$.
Let $\sigma^{slide}_{p',p''}$ be a strategy that reveals no information whenever the belief is in the interval $[p',p'')$.

We now compare the time it takes for each of the strategies $\sigma^{split}_{p',p''}$
and $\sigma^{slide}_{p',p''}$ to make the receiver's belief move from $p'$ to $p''$.
Denote by $\tau_{p''} := \min\{t \geq 0 \colon p^t = p''\}$ 
the first time when the receiver's belief is $p''$,
by $Y^{split}_{p',p''} := \mathbb{E}_{\sigma^{split}_{p',p''}}\left[1-e^{-r\tau_{p''}}\right]$ the expected discounted time to reach belief $p''$ from belief $p'$ under $\sigma^{split}_{p',p''}$,
and by $Y^{slide}_{p',p''}$ the corresponding quantity under $\sigma^{slide}_{p',p''}$.

Under the strategy $\sigma^{split}_{p',p''}$, when $p^0 = p'$,
the stopping time $\tau_{p''}$ has exponential distribution with parameter $\Lambda:=\frac{\lambda_0-p'(\lambda_0+\lambda_1)}{p''-p'}$. 
Using the definition of $p^*$ and defining $\mu:=\frac{r}{\lambda_0+\lambda_1}$, simple algebraic manipulations yield that
\begin{equation}
\label{equ:time}
    Y^{split}_{p',p''}=\frac{\mu\cdot(p''-p')}{p^*-p'+\mu\cdot(p''-p')}.
\end{equation}
This in turn implies that
\begin{equation}
\label{equ:split1}
    w(p',\sigma^{split}_{p',p''})=Y^{split}_{p',p''}\cdot {u}(p') + (1-Y^{split}_{p',p''})\cdot w(p'', \sigma^{split}_{p',p''}).
\end{equation}

When $p' < p^* < p''$ or $p'' < p^* < p'$,
under the strategy $\sigma^{slide}_{p',p''}$, when $p^0 = p'$,
the belief will converge to $p^*$ and never reach $p''$.
For the following computation we therefore assume that $p' < p'' < p^*$
(or, analogously, $p^* < p'' < p'$).
Under $\sigma^{slide}_{p',p''}$, when $p^0 = p'$, we have

\begin{equation}\label{equ slide time} \tau_{p''} = -\frac{1}{\lambda_0+\lambda_1}\ln\left(\frac{p^*-p''}{p^*-p'}\right), 
\end{equation}
hence 
\begin{equation}
\label{equ:time:discounted}
Y^{slide}_{p',p''}=
    1-\left(\frac{p^*-p''}{p^*-p'}\right)^{\frac{r}{\lambda_0+\lambda_1}}= 1-\left(\frac{p^*-p''}{p^*-p'}\right)^{\mu},
\end{equation}
and a relation analogous to Eq.~\eqref{equ:split1} holds.

\begin{lemma}
 $Y^{slide}_{p',p''}\geq Y^{split}_{p',p''}$ for every $p' < p'' \leq p^*$ or $p'>p''\geq p^*$. 
\end{lemma}

\begin{proof}
By Eqs.~\eqref{equ:time} and~\eqref{equ:time:discounted}, we need to show that 

\[1-\left(\frac{p^*-p''}{p^*-p'}\right)^\mu \geq \frac{\mu\left(1-\frac{p^*-p''}{p^*-p'}\right)}{1+\mu\left(1-\frac{p^*-p''}{p^*-p'}\right)}.\]

Denoting $k:=\frac{p^*-p''}{p^*-p'}\in(0,1)$, we need to show that

\[1-k^\mu \geq \frac{\mu(1-k)}{1+\mu(1-k)}.\]

Simple algebraic manipulations show that the above inequality is equivalent to

\[1+\mu k^{\mu+1}-(\mu+1)k^\mu\geq 0.\]

For $k=0$ the inequality is strict and for $k=1$ it is weak. 
Finally, the derivative of the left-hand side with respect to $k$ is negative. This completes the proof.
\end{proof}

\begin{conclusion}\label{conc: linear}
Let $\sigma$ be a Markovian sender's message strategy,
where the belief is split between $p'$ and $p''$ for every $p\in[p',p'']$,
where
$p'\leq p^* \leq p''$. 
Then the resulting putative value function is 
\begin{eqnarray}\nonumber
w(p,\sigma)&=&u(p')\cdot\frac{p'' \cdot(\mu+1) - p^*}{(p''-p')(\mu+1)}+u(p'')\cdot\frac{p^*-p'\cdot(\mu+1)}{(p''-p')(\mu+1)}\\ \label{equ:conc:1}
&&+ p\mu\cdot\frac{ u(p'')-u(p')}{(p''-p')(\mu+1)}.
\end{eqnarray}
\end{conclusion}

Conclusion~\ref{conc: linear}
is obtained by using Eq.~\eqref{equ:split1} twice, once 
as it appears, and once with the roles of $p'$ and $p''$ exchanged.
By definition, $w(\cdot,\sigma)$ is linear on $[p',p'']$, and 
Eq.~\eqref{equ:conc:1} 
follows.

\subsection{Monotonicity of the Value Function}\label{sec monoton}

It is well known that
the value functions
in $G_{cont}(u)$ and $G_\Delta(u)$ 
are
concave, whether or not $u$ is continuous.
In this section we explore 
monotonicity
properties of the value function
under the assumption that $u$ is continuous.
We argue that when $u$ is nondecreasing,
$v_{cont}$ is nondecreasing as well, 
and if $u(1)>u(p^*)$, then the value function is strictly increasing.
Recall that when $u$ is continuous, by  \cite{cardaliaguet_markov_2016} 
there is an optimal message strategy that is Markovian.

\begin{lemma}
\label{lemma:monotonicity:v}
Suppose that the indirect payoff function $u$ is 
continuous and
nondecreasing.
Then $v_{cont}$ is nondecreasing.
If 
$u(1)>u(p^*)$, 
then $v_{cont}$ is increasing.\footnotemark
\end{lemma}

\footnotetext{
Lemma~\ref{lemma:monotonicity:v} is valid even without the assumption that $u$ is continuous,
yet we will use it only for the approximating functions
$(\overline u_\delta)_{\delta > 0}$,
which are continuous.
}

\begin{proof}
If $p^* = 1$, then $v_{cont}(1) = u(1)$.
Since $u$ is nondecreasing, we moreover have $v_{cont}(p) \leq u(1)$ for every $p \in [0,1)$.
This implies that the maximum of $v_{cont}$ is attained at 1, 
and the concavity of $v_{cont}$ implies that $v_{cont}$ is nondecreasing.

Assume then that $p^* < 1$.
Since $v_{cont}$ is concave, to prove that it is nondecreasing it is sufficient to verify that it is nondecreasing on $[p^*,1]$.
Let $p^* \leq p_1 < p_2$.
We will prove that $v_{cont}(p_1) \leq v_{cont}(p_2)$.
We distinguish between two cases: $v_{cont}(p_1) \leq u(p_2)$ and $v_{cont}(p_1) > u(p_2)$.

\bigskip
\noindent\textbf{Case 1: $v_{cont}(p_1) \leq u(p_2)$.}
Suppose that the initial belief is $p_2$, and consider a message strategy $\sigma$ that splits the belief of the receiver between $p_1$ and $p_2$ for all beliefs in $(p_1,p_2]$, and plays optimally once the belief is $p_1$.
As long as $p^t = p_2$, the instantaneous payoff is $u(p_2) \geq v_{cont}(p_1)$, 
and once the belief is $p^t = p_1$, the continuation payoff is $v_{cont}(p_1)$. 
Therefore, 
\[ v_{cont}(p_2) \geq w(p_2,\sigma) \geq v_{cont}(p_1). \]

\bigskip
\noindent\textbf{Case 2: $v_{cont}(p_1) > u(p_2)$.}

Let $\sigma$ be an optimal message strategy in $G_{cont}(u)$, 
so that
\begin{equation}
    \label{equ:83}
u(p_2) < v_{cont}(p_1) = w(p_1,\sigma). 
\end{equation}
Let $\tau$ be the first time when $p^\tau \geq p_2$.
Since $u$ is monotone, 
when the initial receiver's belief is $p_1$,
until time $\tau$, the instantaneous payoff is at most $u(p_2) < v_{cont}(p_1)$.
Since $w(p_1,\sigma)$ is a convex combination of the payoff until time $\tau$
and the payoff after time $\tau$, 
Eq.~\eqref{equ:83} implies that $\prob_{p_1,\sigma}(\tau < \infty) > 0$ and
\[  v_{cont}(p_1) = w(p_1,\sigma) <  \E_{p_1,\sigma}[v_{cont}(p^\tau) \mid \tau < \infty]. \]
Thus, there is a belief $p_3 \geq p_2$ such that $v_{cont}(p_3) > v_{cont}(p_1)$.
The concavity of $v_{cont}$ implies that $v_{cont}(p_2) > v_{cont}(p_1)$.

\bigskip

We turn to prove the second claim.
Assume that $u(p^*) < u(1)$,
and suppose, by way of contradiction, that $v_{cont}$ is not increasing.
Since $v_{cont}$ is concave and nondecreasing,
this implies that there is $p' \in (p^*,1)$ such that $v_{cont}$ is constant on $[p',1]$.
Since $u(p^*) < u(1)$, since $u(p) \leq u(1)$ for every $p \in [0,1]$, and since $u$ is continuous, 
we have $v_{cont}(p') < u(1)$. 
Let $\sigma$ be a sender's strategy that splits the belief of the receiver between 1 and $p'$ for all beliefs in $(p',1)$, and plays optimally once the belief reaches $p'$. 
Then
\[ u(1) \geq  v_{cont}(1) \geq w(1,\sigma) = Y^{split}_{1,p'}\cdot u(1) + (1-Y^{split}_{1,p'})\cdot v_{cont}(p'). \]
Since $Y^{split}_{1,p'}$ is positive,
this implies that $v_{cont}(1) > v_{cont}(p')$, a contradiction.
\end{proof}

%As a conclusion of the concavity and %monotonicity of $v_{cont}$ we deduce the following.

%\begin{corollary}
%\label{corollary:v'}
%Suppose that the indirect payoff function $u$ is nondecreasing.
%Then whether or not $u$ is continuous, 
%the right-derivative of $v_{cont}$ is nonincreasing.
%\end{corollary}

\subsection{Strategies in Continuous Time --- Previous Results}\label{sec:str cont}

\cite{cardaliaguet_markov_2016} studied our game when the indirect payoff function $u$ is continuous,
characterized the value function,
proved that the 
sender has an optimal message strategy,
and characterized such a strategy. 
\cite{gensbittel_2020} further studied the game when $u$ is continuous.
In this section we present two results from these papers.

Recall that the \textit{hypograph} of a function
$f : [0,1] \to \mathbb{R}$ is the set of all points that lie on or below the graph of the function.
When $f$ is concave, its hypograph is a convex set, and its set of extreme points coincides with the set of points on the graph of $f$ where $f$ is not affine,
plus the corner points $(0,f(0))$ and $(1,f(1))$.

For simplicity of presentation, define 
\[ \mu:=\frac{r}{\lambda_0+\lambda_1}. \]
This is the ratio between the discount rate and the rate at which the state changes.

\begin{Theorem}[Theorem~2.12 in \cite{gensbittel_2020}, and Theorem~2.3 in \cite{ashkenazi-golan_solving_2020}]
\label{Char}
Provided the indirect payoff function $u$ is continuous,
the value function $v_{cont}$ in $G_{cont}(u)$ is the unique continuous, concave function $v : [0,1]\to\mathbb{R}$ that is differentiable on $[0,1]$, except, possibly,
at $p^*$, and satisfies the following conditions:
\begin{itemize}
\item[G.1] $v_{cont}(p^*)\geq u(p^*)$, with equality if $(p^*,v_{cont}(p^*))$ is an extreme point of the hypograph of $v_{cont}$.
\item[G.2] For every $p\in[0,1]\setminus\{ p^*\}$ we have $v'(p)(p-p^*) + \mu\cdot\left(v_{cont}(p)-u(p)\right)\geq 0$.
\item[G.3] For every extreme point $(p,v_{cont}(p))$ of the hypograph of $v_{cont}$ such that $p \neq p^*$ we have
\begin{equation}
\label{eqv}
v'(p)(p-p^*) + \mu\cdot\left(v_{cont}(p)-u(p)\right)= 0,
\end{equation}
\end{itemize}
where for $p=0$ (resp.~$p=1$), $v'(p)$ stands for the right (resp.~left) derivative of $v_{cont}$ at $p$.
\end{Theorem}

Observe that points $(p,v_{cont}(p))$ that are not extreme points of the hypograph of $v_{cont}$ lie on a line segment connecting two extreme points of the hypograph of; that is, 
they
are convex combinations of these two extreme points,
denoted $(p',v_{cont}(p'))$ and $(p'',v_{cont}(p''))$. 
This implies that the value at such belief $p$ can be obtained by a split of the belief between $p'$ and $p''$.

We will use the above Theorem~\ref{Char} to obtain the optimal message strategy for beliefs outside the continuity interval $[p_0,p_1]$. 
For beliefs in the continuity interval $[p_0,p_1]$ we use the following characterization of the optimal message strategy. This result follows from \cite{cardaliaguet_markov_2016} and applies to both
$G_{cont}(u)$ and $G_{cont}(\overline u_\delta)$.

\begin{lemma}
\label{lemma:CRRV}
If $p^*$ is a discontinuity point of $u$ (so that $p^* = p_0$),
then the sender's optimal message strategy  at receiver's belief $p^*$ for both $u$ and $\overline{u}_\delta$, is to reveal no information. 

If $p^* \in (p_0,p_1)$, then for both $u$ and $\overline{u}_\delta$, for every $p \in [p_0,p_1]$,  the optimal message strategy at receiver's belief $p$ is to split the belief between $p_0$ and $p_1$. 
\end{lemma}

\bigskip

\begin{proof}
The result follows from Lemma~3 in \cite{cardaliaguet_markov_2016},
which states that if $(p^*,(\cav  u)(p^*))$ lies on the line segment that connects $(p',  u(p'))$ and $(p'', u(p''))$,
for some $p',p'' \in [0,1]$ that satisfy 
$p' \leq p^* \leq p''$, 
then the value function is linear on $[p',p'']$, and the optimal message strategy at each belief $p \in [p',p'']$ is to split the belief between $p'$ and $p''$
(and to reveal no information if $p' = p'' = p^*$).

Since $u$ has a concave envelope, 
if $p^* = p_0$, then $ u(p) = (\cav  u)(p)$,
and then the result follows by setting $p' = p'' := p_0$.
If $p^* \in (p_0,p_1)$, then the result follows 
by setting $p' := p_0$ and $p'' := p_1$. The same reasoning holds for $\overline{u}_{\delta}$.

While \cite{cardaliaguet_markov_2016} analyze a model where $u$ is continuous, their Lemma~3 does not depend on the continuity of $u$.

\end{proof}

The intuition behind Lemma~\ref{lemma:CRRV} is as follows.
%Denote by $p^t$ the receiver's belief at time $t$;
%that is, the probability that the receiver assigns to the event that the state at time $t$ is 1.
When the initial belief $p^0$ is $p^*$, 
for every message strategy the unconditional expectation $\E[p^t]$ is equal to $p^*$.
The expected instantaneous payoff is $\E[u(p^t)]$,
which, by Jensen's inequality, is smaller than $(\cav u)(\E[p^t]) = (\cav u)(p^*)$.

Consider now the message strategy $\sigma^*$ described in Theorem~\ref{theorem:main_result}.
If $p^* = p_0$, then $p^t = p_0$ for every $t \geq 0$ and $(\cav u)(p^*) = u(p^*)$.
It follows that the sender's payoff under $\sigma^*$ is $(\cav u)(p^*)$,
which is the best possible payoff.
If $p^* \in (p_0,p_1)$, then $p' = p_0$ and $p'' = p_1$, and
the posterior belief $p^t$ is either $p'$ or $p''$: when, say, $p^t = p''$, the Markov transition makes the belief slide toward $p^*$,
and then the sender splits the belief again between $p'$ and $p''$. Since the unconditional expectation of $p^t$ is $p^*$,
the unconditional probability $\alpha$ that the belief at period $n$ is $p'$ satisfies
$\alpha p_0 + (1-\alpha) p_1 = p^*$.
As a result,
this message strategy guarantees to the sender the payoff $\alpha u(p') + (1-\alpha) u(p'')$, which is equal to $(\cav u)(p^*)$.
Hence, 
in this case as well, $\sigma^*$ guarantees to the sender the highest possible payoff.

The above discussion provides the optimal message strategy for the continuity interval $[p_0,p_1]$. 
In the next 
subsections we handle
the other continuity intervals.

\subsection{The Functions $(\overline g_\delta)_{\delta > 0}$\label{sec g}}

Inspired by Theorem~\ref{Char},
for every $\delta > 0$ sufficiently small
define 
a function $\overline g_\delta : [0,1] \setminus \{p^*\} \to \dR$ by

\[ \overline g_\delta(p):=\mu\cdot\left(\frac{\overline u_\delta(p)-\overline v_\delta(p)}{p-p^*}\right), \ \ \ \forall p \in [0,1] \setminus \{p^*\}. \]
In this section we will study the function $\overline g_\delta$.
Since $\overline u_\delta$ is continuous and $\overline v_\delta$ is Lipshitz, 
$\overline g_\delta$ is continuous.

By Theorem~\ref{Char}, 
the function $\overline g_\delta$ is related to the derivative of $\overline v_\delta$.
Indeed, by (G.2), $\overline v'_\delta(p)\leq \overline g_\delta(p)$ for every $p<p^*$, 
and  $\overline v'_\delta(p)\geq \overline g_\delta(p)$ for every $p>p^*$. 
By (G.3),  
$\overline g_\delta(p)=\overline v'_\delta(p)$ for $p\neq p^*$ such that $(p,\overline v_\delta(p))$ is an extreme point of the hypograph of $v_{cont}$. 
Furthermore, for $p\neq p^*$ the function $\overline v_\delta(p)$ is linear over segments $[p',p'']$ where all $p\in (p',p'')$ are not extreme points of the hypograph of $\overline v_\delta(p)$. 
Hence, the derivative $\overline v'_\delta$ is constant over such interval $(p',p'')$ and 
satisfies
$\overline g_\delta(p')=\overline g_\delta(p'')$.
We deduce the following.

\begin{lemma}\label{lemma same der}
If $(p', \overline v_\delta (p'))$ and $(p'', \overline v_\delta (p''))$ are extreme points of the hypograph of $\overline v_\delta$, and none of the points $(p,\overline v_\delta(p))$,  for $p\in (p',p'')$, is an extreme point of the hypograph of $\overline v_\delta$, then $\overline g_\delta(p')=\overline g_\delta(p'')$.
Moreover, if $p''<p'\leq p^*$ or $p''>p'\geq p^*$, then 
\[\overline v_\delta(p'')=\frac{\mu\cdot(p'-p'')}{p^*-p''+\mu\cdot(p'-p'')}\cdot\overline u_\delta(p'')+\frac{p^*-p''}{p^*-p''+\mu\cdot(p'-p'')}\cdot\overline v_\delta(p').\]
\end{lemma}

\begin{proof}
The first claim holds since $\overline v_\delta$ is linear on $[p',p'']$.
From this we conclude that
\begin{equation}
    \frac{\overline v_\delta(p'')-\overline v_\delta(p')}{p''-p'}=\frac{\mu\cdot\left(\overline u_\delta(p'')-\overline v_\delta(p'')\right)}{p''-p^*}.
    \label{equ:11}
\end{equation}
The second claim follows 
from Eq.~\eqref{equ:11} and
by simple algebraic manipulations.
\end{proof}

\begin{comment}
Lemma~\ref{lemma same der} holds when $p'' < p'$ as well,
and then the interval is $(p'',p')$.
\end{comment}

\begin{comment}
The results in this subsection hold whenever the indirect payoff $u$ is continuous.
\end{comment}

The next result describes the graph of $\overline g_\delta(p)$ on the segments $[0,p_0]$ and $[p_1,1]$.
It's proof is not inspiring and is relegated to Section~\ref{section:proof:lemma g}.
Note that we do not%
\footnote{When $p^*\in(p_0,p_1)$,
simple computations yield that $\overline{g}_\delta(p) =\frac{h_1-h_0}{(p_1-p_0)(\mu+1)}\left(\frac{(p_0-p^*)(\mu+1)}{p-p^*}-\mu\right)$ for $p\in(p_0, p_1-\delta)$. 
Note that in this case $\overline{g}_\delta$ is increasing on both segments $(p_0,p^*)$ and $(p^*,p_1)$.} 
describe $\overline g_\delta$ on
$[p_0,p_1]$ if $p^*>p_0$.

\begin{lemma}\label{lemma g}
For every $\delta > 0$ sufficiently small, the function $\overline g_\delta$
satisfies the following properties, see Figure~\ref{fig:g arrow}:
\begin{enumerate}
    \item[(a)] $\overline g_\delta$ increases on $(p_{-j-1}, p_{-j}-\delta)$, for $j\in\left\{ 0 ,\ldots,m-1\right\}$.
    
    \item[(b)] $\overline g_\delta$ decreases on $(p_{-j}-\delta, p_{-j})$,  for
    $j\in\left\{1,\ldots,m-1\right\}$ (if $p^* = p_0$) or 
    $j\in\left\{0,1,\ldots,m-1\right\}$ (if $p^*>p_0$).
    \item[(c)] If $p^*=p_0$, then:
    \begin{enumerate}
        \item[(i)] $\overline g_\delta$ increases on $(p_0-\delta ,p_0)$, and
        \item[(ii)] $\overline g_\delta$ is smaller or equal%
        \footnote{We cannot determine whether it increases or decreases on this interval using simple observations like is done in this lemma. Later on we will be able to conclude that it is actually constant on this interval.}
        to $\overline{g}_\delta(p_1)$ on $(p_0, p_1-\delta)$.
    \end{enumerate}
    \item[(d)]
    For each $j\in\left\{1,\ldots  m'-1\right\}$ 
    there is $d_j \in (p_j,p_{j+1}-\delta)$ such that 
    $\overline g_\delta$ 
    is positive and decreasing on $[p_j, p_j +d_j)$, and if it is zero at $p_j+d_j$, then it remains nonpositive on $[p_j+d_j,p_{j+1}-\delta)$.
    \item[(e)] $\overline g_\delta$ increases on $(p_j-\delta,p_j)$, for $j\in\left\{1,\ldots  m'\right\}$.
\end{enumerate}

\end{lemma}

\begin{figure}[h]
\centering
\begin{tikzpicture}[domain=0:0.7,xscale=4,yscale=2.8]
\draw[ultra thick] (0,0)--(3.2,0);
\draw[ultra thick] (0,0)--(0,-0.1) node[below]{$p_{-m}$};
\draw[] (0.2,0)--(0.2,-0.3) node[below]{$p_{-m+1}-\delta$};
\draw[ultra thick] (0.4,0)--(0.4,-0.1) node[below]{$p_{-m+1}$};
\draw[] (0.6,0)--(0.6,-0.4) node[below]{$p_{-m+2}-\delta$};
\draw[ultra thick] (0.8,0)--(0.8,-0.1) node[below]{$p_{-m+2}$};
\draw[] (1.2,0)--(1.2,-0.3) node[below]{$p_{-1}-\delta$};
\draw[ultra thick] (1.4,0)--(1.4,-0.1) node[below]{$p_{-1}$};
\draw[] (1.6,0)--(1.6,-0.3) node[below]{$p_{0}-\delta$};
\draw[ultra thick] (1.8,0)--(1.8,-0.1) node[below]{$p_{0}$};
\draw[ultra thick] (2,0)--(2,-0.1) node[below]{$p_{1}$};
\draw[] (2.1,0)--(2.1,-0.5) node[below]{$p_{2}+d_1$};
\draw[] (2.3,0)--(2.3,-0.3) node[below]{$p_{2}-\delta$};
\draw[ultra thick] (2.4,0)--(2.4,-0.1) node[below]{$p_{2}$};
\draw[] (2.5,0)--(2.5,-0.5) node[below]{$p_{2}+d_2$};
\draw[] (2.7,0)--(2.7,-0.3) node[below]{$p_{3}-\delta$};
\draw[ultra thick] (2.8,0)--(2.8,-0.1) node[below]{$p_{3}$};
\draw[ultra thick] (3.2,0)--(3.2,-0.1) node[below]{$p_{m'}$};

\draw [-to] (0.05,0.05)--(0.15,0.2); 
\draw [-to] (0.25,0.2)--(0.35,0.05); 
\draw [-to] (0.45,0.05)--(0.55,0.2); 
\draw [-to] (0.65,0.2)--(0.75,0.05); 

\draw [-to] (1.25,0.2)--(1.35,0.05); 
\draw [-to] (1.45,0.05)--(1.55,0.2); 
\draw [-to] (1.65,0.05)--(1.75,0.2); 

\draw [-to] (2.05,0.2)--(2.15,0.05); 
\draw [-to] (2.25,0.05)--(2.35,0.2); 
\draw [-to] (2.45,0.2)--(2.55,0.05); 
\draw [-to] (2.65,0.05)--(2.75,0.2); 
\draw [fill] (.9,.1) circle [radius=0.01];
\draw [fill] (1,.1) circle [radius=0.01];
\draw [fill] (1.1,.1) circle [radius=0.01];
\draw [fill] (2.9,.1) circle [radius=0.01];
\draw [fill] (3,.1) circle [radius=0.01];
\draw [fill] (3.1,.1) circle [radius=0.01];
\node at (1.9,.13) {?};
\end{tikzpicture}
\caption{The regions where $\overline g_\delta$ increases and decreases.} \label{fig:g arrow}
\end{figure}
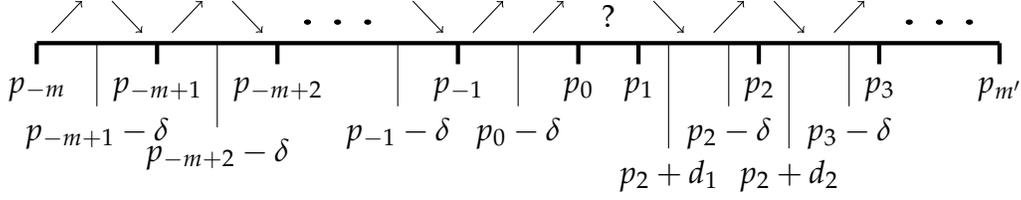

%Observe also on the segments $[0,p^*)$ and $(p^*,1]$, we have that $\overline g_\delta(p)\geq 0$. This is because for $p<p^*$ the numerator of $g$ is nonpositive and the denominator is negative, and for $p>p^*$ the numerator is non-negative and the denominator is positive. 

To complete the description of the function $\overline g_\delta$ we compare the values that $\overline g_\delta$ attains at the discontinuity points of $u$.   

\begin{lemma}\label{lemma g at endpoints}
For every $\delta > 0$ sufficiently small,
the function $\overline g_\delta$ satisfies the following properties:
\begin{enumerate}
    \item[(a)]  $\overline g_\delta(p_{-j})>\overline g_\delta(p_{-j+1})$ for $j\in\{2,\ldots, m\}$, and if $p^*>p_0$, then $\overline g_\delta(p_{-1})>\overline g_\delta(p_0)$.
    \item[(b)] $\overline g_\delta(p_{j})>\overline g_\delta(p_{j+1})$ for $j\in\{1,\ldots,m'-1\}$.
    \item[(c)] If $p^*=p_0$, then
    \begin{enumerate}
        \item[(i)] $\overline g_\delta(p_{-1})<\lim_{\eta\to 0}\overline g_\delta(p_0-\eta)$, and
        \item[(ii)] $\overline g_\delta(p_1)\geq\lim_{\eta\to 0}\overline g_\delta(p_0+\eta)$.
    \end{enumerate}
\end{enumerate}
\end{lemma}

The proof of Lemma~\ref{lemma g at endpoints} is relegated to Section~\ref{sec proof g endpoints}.
Figures~\ref{fig:g left} and~\ref{fig:g right} summarize Lemmas~\ref{lemma g} and~\ref{lemma g at endpoints}.
In these figures, 
the graph of the function $\overline g_\delta$ is the dashed line.
The continuity of $\overline g_\delta$ on $[0,1] \setminus \{p^*\}$
ensures 
that for every $j\in\{1,\ldots,m-1\}$ there exists $\overline{q}_{-j}(\delta)\in(p_{-j}-\delta,p_{-j})$ such that $\overline g_\delta(\overline{q}_{-j}(\delta))=\overline g_\delta(p_{-j})$, and if $p^*>p_0$,
then this conclusion
holds for $j=0$ as well, see Figure~\ref{fig:g left}. 
Similarly, for every $j\in\{1,\ldots,m'-1\}$  there exists $\overline{q}_j(\delta)\in(p_j, p_{j+1}-\delta)$ such that $\overline g_\delta(\overline{q}_j(\delta))=\overline g_\delta(p_{j+1})$, see Figure~\ref{fig:g right}.
Note that the function $\overline g_\delta$ is \textit{not} piecewise constant.
%;
%the direction of the arrows in Figures~\ref{fig:g left} and~\ref{fig:g right} indicates the regions where the function is increasing or decreasing.

\begin{figure}[h]
\centering
\begin{tikzpicture}[domain=0:0.7,xscale=17,yscale=2.8]
\draw[<->] (0,2)-- (0,0) -- (.5,0) node[below] {$p$};
%\draw[dashed,->] (0.02,1.2)--(0.12,1.75);
%\draw[dotted, blue] (0.02,1.2)--(0.137,1.2);
\draw[dotted, blue] (0.15,1.2)--(0.15,-0.3) node[below]{$\overline{q}_{-j}(\delta)$};
\draw[dashed] (0.02,1.2)--(0.13,1.8)--(0.16,1)--(0.26,1.4)--(0.295,0.6)--(0.4,1) node[above]{$\overline g_\delta$} ;
%\draw[dashed,->] (0.13,1.8)--(0.145,0.95);
%\draw[dashed,->] (0.16,1)--(0.25,1.4) ;
\draw[dotted, blue] (0.28,0.95)--(0.28,-0.3) node[below]{$\overline{q}_{-j+1}(\delta)$};
%\draw[dotted, blue] (0.3,0.57)--(0.405,0.57);
\draw[dotted, blue] (0.41,0.57)--(0.41,-0.3) node[below]{$\overline{q}_{-j+2}(\delta)$};
%\draw[dashed,->] (0.4,1)--(0.415,0.35);
\draw[dashed] (0.4,1)--(0.415,0.35);
\draw[] (0.02,0.02)--(0.02,-0.02) node[below]{$p_{-j-1}$};
\draw[] (0.12,0.02)--(0.12,-0.02);
\draw[] (0.16,0.02)--(0.16,-0.02) node[below]{$p_{-j}$};
\draw[] (0.257,0.02)--(0.257,-0.02);
\draw[] (0.295,0.02)--(0.295,-0.02) node[below]{$p_{-j+1}$};
\draw[] (0.38,0.02)--(0.38,-0.02);

\draw[] (0.415,0.02)--(0.415,-0.02) node[below]{$p_{-j+2}$};
\draw[thick, red] (0.02,1.19)--(0.150,1.19); 
\draw[thick, red] (0.150,1.19)--(0.16,0.97); 
\draw[thick, red] (0.16,0.97)--(0.28,0.97); 
\draw[thick, red] (0.28,0.97)--(0.295,0.58); 
\draw[thick, red] (0.295,0.58)--(0.41,0.58); 
\draw[thick, red] (0.41,0.58)--(0.415,0.37); 
\draw[thick, red] (0.2,0.94) node[anchor=north]{$\overline v'_\delta$}; 
\draw[dotted,-] (0.02,1.19)--(0.02,0);
\draw[dotted,-] (0.16,0.94)--(0.16,0);
\draw[dotted,-] (0.295,0.56)--(0.295,0);
\draw[dotted,-] (0.415,0.37)--(0.415,0);

\draw [pen colour={blue},
    decorate, 
    decoration = {calligraphic brace,
        raise=2pt,
        amplitude=5pt}] (0.12,0) --  (0.16,0)
    node[pos=0.5,above=7pt,blue]{$\delta$};
\draw [pen colour={blue},
    decorate, 
    decoration = {calligraphic brace,
        raise=2pt,
        amplitude=5pt}] (0.255,0) --  (0.295,0)
    node[pos=0.5,above=7pt,blue]{$\delta$};
\draw [pen colour={blue},
    decorate, 
    decoration = {calligraphic brace,
        raise=2pt,
        amplitude=5pt}] (0.38,0) --  (0.415,0)
    node[pos=0.5,above=7pt,blue]{$\delta$};
\end{tikzpicture}
\caption{The functions $\overline g_\delta$ (dashed) and $\overline v'_\delta$ (red) for $p< p_0$.} \label{fig:g left}
\end{figure}
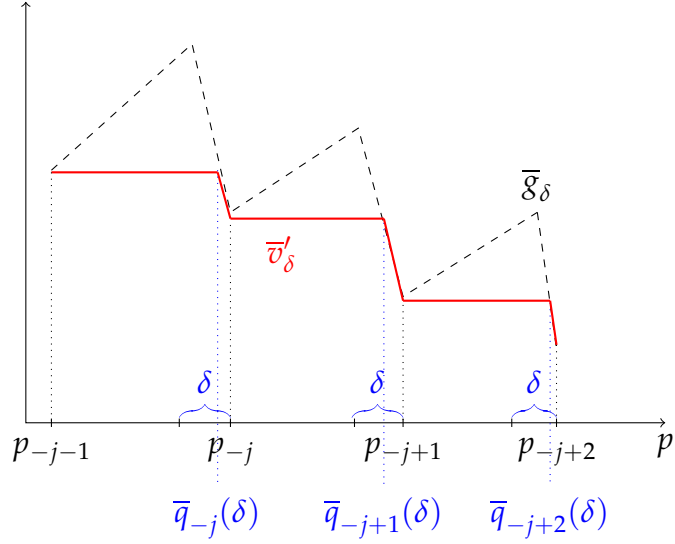

%%%

\begin{figure}[h]
\centering
\begin{tikzpicture}[domain=0:0.7,xscale=17,yscale=2.8]
\draw[<->] (0,2)-- (0,0) -- (.5,0) node[below] {$p$};
\draw[dashed] (0.02,1.75)--(0.13,1.2)--(0.15,1.5)--(0.26,0.95) node[below]{$\overline g_\delta$};
%\draw[dashed,->] (0.02,1.75)--(0.12,1.2);
%\draw[dashed,->] (0.13,1.2)--(0.145,1.5);
\draw[dotted, blue] (0.07,1.5)--(0.145,1.5);
\draw[dotted, blue] (0.07,1.5)--(0.07,-0.3) node[below]{$\overline{q}_{j-1}(\delta)$};
%\draw[dashed,->] (0.15,1.5)--(0.25,0.95) node[below]{$\overline g_\delta$};
%\draw[dashed,->] (0.26,0.95)--(0.28,1.35);
\draw[dashed] (0.26,0.95)--(0.285,1.35)--(0.4,0.8)--(0.41,1.1);
\draw[dotted, blue] (0.18,1.35)--(0.285,1.35);
\draw[dotted, blue] (0.18,1.35)--(0.18,-0.3) node[below]{$\overline{q}_{j}(\delta)$};
%\draw[dashed,->] (0.285,1.35)--(0.39,0.8);
%\draw[dashed,->] (0.4,0.8)--(0.41,1.1);
\draw[dotted, blue] (0.33,1.1)--(0.415,1.1);
\draw[dotted, blue] (0.33,1.1)--(0.33,-0.3) node[below]{$\overline{q}_{j+1}(\delta)$};
\draw[thick, red] (0.021,1.75)--(0.067,1.5);
\draw[thick, red] (0.067,1.5)--(0.15,1.5);
\draw[thick, red] (0.15,1.5)--(0.18,1.35);
\draw[thick, red] (0.18,1.35)--(0.287,1.35);
\draw[thick, red] (0.287,1.35)--(0.335,1.1);
\draw[thick, red] (0.335,1.1)--(0.41,1.1);
\draw[thick, red] (0.25,1.35) node[anchor=south]{$\overline v'_\delta$}; 

\draw[] (0.02,0.02)--(0.02,-0.02) node[below]{$p_{j-1}$};
\draw[] (0.12,0.02)--(0.12,-0.02);
\draw[] (0.15,0.02)--(0.15,-0.02) node[below]{$p_{j}$};
\draw[] (0.25,0.02)--(0.25,-0.02);
\draw[] (0.28,0.02)--(0.28,-0.02) node[below]{$p_{j+1}$};
\draw[] (0.41,0.02)--(0.41,-0.02) node[below]{$p_{j+2}$};

\draw[dotted,-] (0.02,1.75)--(0.02,0);
\draw[dotted,-] (0.15,1.5)--(0.15,0);
\draw[dotted,-] (0.28,1.35)--(0.28,0);
\draw[dotted,-] (0.41,1.1)--(0.41,0);

\draw [pen colour={blue},
    decorate, 
    decoration = {calligraphic brace,
        raise=2pt,
        amplitude=5pt}] (0.12,0) --  (0.15,0)
    node[pos=0.5,above=7pt,blue]{$\delta$};
\draw [pen colour={blue},
    decorate, 
    decoration = {calligraphic brace,
        raise=2pt,
        amplitude=5pt}] (0.25,0) --  (0.28,0)
    node[pos=0.5,above=7pt,blue]{$\delta$};
\draw [pen colour={blue},
    decorate, 
    decoration = {calligraphic brace,
        raise=2pt,
        amplitude=5pt}] (0.38,0) --  (0.41,0)
    node[pos=0.5,above=7pt,blue]{$\delta$};
\end{tikzpicture}
\caption{The functions $\overline g_\delta$ (dashed) and $\overline v'_\delta$ (red) for $p>p_1$.} \label{fig:g right}
\end{figure}
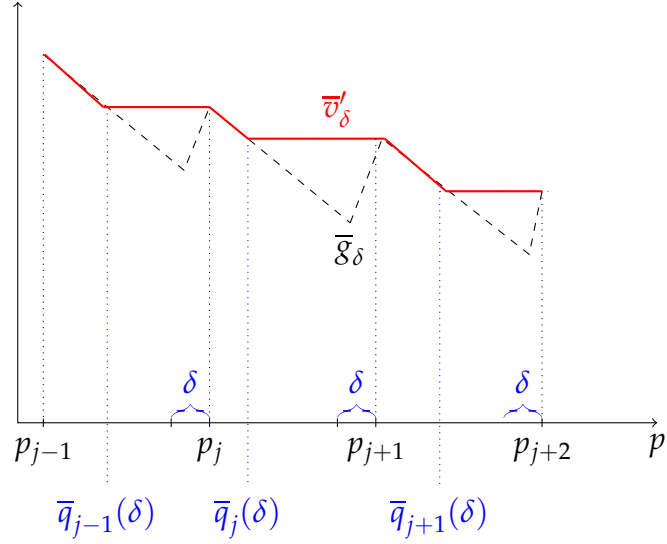
%%%

\subsection{The Derivative of the Value Function $\overline v_\delta$}\label{sec der}

Theorem~\ref{Char} and the results so far allow us to
describe the structure of $\overline v_\delta$, and specifically, we focus on its derivative.
The value function $\overline v_\delta$ is concave, 
and by Lemma~\ref{lemma:monotonicity:v} it is nondecreasding.
Hence, $\overline v'_\delta$ is nonnegative and nonincreasing.
By Theorem~\ref{Char}(G.2),
$\overline v'_\delta \leq \overline g_\delta$ on $[0,p^*)$,
and 
$\overline v'_\delta \geq \overline g_\delta$ on $(p^*,1]$.
In intervals where $\overline v'_\delta$ is constant, $\overline v_\delta$ is linear, and the two endpoints of such intervals are extreme points of the hypograph of $\overline v_\delta$. If there exists no $\varepsilon>0$ such that $\overline{v}'_\delta$ is constant on $(p-\varepsilon,p+\varepsilon)$,
then $(p,\overline v_\delta(p))$ is an extreme point of the hypograph of $\overline v_\delta$,
hence $\overline v'_\delta(p) = \overline g_\delta(p)$
(Theorem~\ref{Char}(G.3)).

The unique function that satisfies these properties is the function that is displayed in red
in Figures~\ref{fig:g left} and~\ref{fig:g right}:
\begin{itemize}
\item
Since $(0,\overline v_\delta(0))$ is an extreme point of the hypograph of $\overline v_\delta$,
we have $\overline v'_\delta(0) = \overline g_\delta(0)$.
\item
On the interval $[0,\overline q_{-m+1}(\delta)]$ the function $\overline g_\delta$ is at least $\overline v'_\delta(0)$,
hence $\overline v'_\delta$ must be constant on this interval.
\item
On the interval $[\overline q_{-m+1}(\delta),p_{-m+1}]$,
the only function that is 
(a) at most $\overline g_\delta$
and 
(b)
coincides with it when it is not constant, is $\overline g_\delta$.
Hence, $\overline v'_\delta = \overline g_\delta$ on this interval, and so on.
\end{itemize}
Thus,
for every $j\in\left\{2,\ldots, m-1 \right\}$,
and for $j=1$ in case $p^* > p_0$,
\begin{equation}\label{eq der left}
\overline v'_\delta(p)=
\left\{
\begin{array}{lll}
\overline g_\delta(p_{-j}), & \ \ \ \ \ & p\in[p_{-j}, \overline{q}_{-j+1}(\delta)),\\
\overline g_\delta(p), & & p\in[\overline{q}_{-j+1}(\delta),p_{-j+1}].
\end{array}\right.
\end{equation}

For $j=1$ and $p^*=p_0$, we have  $\overline g_\delta(p_{-1})=\overline v'_\delta(p)$ for all $p \in [p_{-1},p_0]$, as discussed in the following remark.

\begin{remark}\label{remark p*} 
For the case where $p_0=p^*$, we need a further observation to describe the value function. 
The function $\overline g_\delta(p)$ is not defined at $p_0=p^*$. The conditions of Theorem~\ref{Char} should hold nonetheless. 
By Lemma~\ref{lemma g}(a) and Lemma~\ref{lemma g}(c.i),
for $p\in(p_{-1},p_0)$ we have $\overline{g}_\delta(p)>\overline{g}_\delta(p_{-1})$. 
Since $\overline v_\delta(p)\leq \overline g_\delta(p)$ for $p\in [p_{-1}, p_0)$, 
and since the derivative of $v_{cont}$ is nonincreasing,
there is no $p \in (p_{-1},p_0)$ 
such that
$\overline{v}'_\delta(p)=\overline g_\delta(p)$. 
We conclude that $(p_{-1},\overline{v}'_\delta(p_{-1}))$ and $(p_0,\overline{v}'_\delta(p_0))$ are extreme points of the hypograph of $\overline v_\delta$. Similar arguments using Lemma~\ref{lemma g}(c.ii) and Lemma~\ref{lemma g}(e) lead to the conclusion that $(p_1,v_{cont}(p_1))$ is an extreme point of that hypograph as well.\end{remark}

Similarly, \begin{itemize}
\item
Since $(p_1,\overline v_\delta(p_1))$ is an extreme point of the hypograph of $\overline v_\delta$,
we have $\overline v'_\delta(p_1) = \overline g_\delta(p_1)$.
\item
On the interval $[p_1,\overline q_{1}(\delta)]$, the only function that is at least $\overline g_\delta$
and coincides with it when it is not constant is $\overline g_\delta$. Therefore, $\overline v'_\delta(p)=\overline g_\delta(p)$ on this interval.
\item
On the interval $[\overline q_{1}(\delta),p_{2}]$,
the function $\overline g'_\delta$ is at most $\overline v'_\delta(\overline q_1(\delta))$ (which is equal to $\overline g_\delta(p_0)$), hence $\overline v_\delta(p)$ is constant on this interval, and so on.
\end{itemize}
Thus,
for every $j\in\left\{1,\ldots, m'-1 \right\}$,
\begin{equation}\label{eq der right}
\overline v'_\delta(p)=
\left\{
\begin{array}{lll}
\overline g_\delta(p), & \ \ \ \ \ & p\in[p_{j}, \overline{q}_{j}(\delta)),\\
\overline g_\delta(p_{j+1}), & & p\in[\overline{q}_{j}(\delta),p_{j+1}].
\end{array}\right.
\end{equation}

%\begin{remark}
%Recall that from Theorem~\ref{Char} (G.3), all points $(p,\overline v_\delta(p))$ such that $\overline v_\delta'(p)=\overline g_\delta(p)$ are extreme points of the hypograph of $\overline v_\delta$. 
%\end{remark}

\subsection{Strategies and the Derivative of the Putative Value they Generate}\label{sec str and der}

Once the derivative of the value function is characterized by Eqs.~\eqref{eq der left} and~\eqref{eq der right}, 
we study the derivative of the putative value generated 
by the strategies $\sigma^{split}_{p',p''}$ and $\sigma^{slide}_{p',p''}$.

Let $\overline{w}_\delta(p,\sigma)$ denote the putative value obtained under sender's message strategy $\sigma$ at belief $p$ in the game $G_{cont}(\overline{u}_\delta)$.

The next result characterizes the derivative of the putative value of $\sigma^{slide}_{p',p''}$ and $\sigma^{split}_{p',p''}$.

\begin{lemma}\label{lemma deriv}
Let $p',p'' \in [0,1]$ be such that either $p' < p'' < p^*$ or $p^* < p'' < p'$.
For every $p$ that lies strictly between $p'$ and $p''$,
\begin{equation}
\label{equ:41}
{\overline{w}_\delta}'(p,\sigma^{slide}_{p',p''})=\mu\cdot\frac{\overline{u}_\delta(p)-{\overline{w}_\delta}(p,\sigma^{slide}_{p',p''})}{p-p^*}
\end{equation} 
and
\begin{equation}
\label{equ:42}
{\overline{w}_\delta}'(p,\sigma^{split}_{p',p''})=\mu\cdot \frac{\overline{u}_\delta(p')-\overline{w}_\delta(p',\sigma^{split}_{p',p''})}{p'-p^*}.
\end{equation} 
If $p=p'$ or $p=p''$, then the directional derivative at $p$
(the left-derivative if $p = \max\{p',p''\}$ or the right-derivative if $p=\min\{p',p''\}$)
is equal to the quantity given above.
\end{lemma}

The proof of Eq.~\eqref{equ:41} involves differentiation of the putative value function, 
and the proof of Eq.~\eqref{equ:42} uses Eqs.~\eqref{equ:time} and~\eqref{equ:split1}.
Both calculations are uninspiring and omitted.

Lemma~\ref{lemma deriv} implies that when $\sigma$ is an optimal strategy, 
if at sender's belief $p$ the sender reveals no information, then
 ${\overline{w}_\delta}'_-(p,\sigma)=\overline g_\delta(p)$.

\subsection{Proof of Lemma~\ref{lemma:structure}}\label{sec lemma cont proof}

In this section we prove Lemma~\ref{lemma:structure} by collecting the results we described so far. 
Let $(\overline q_{-j}(\delta))_{j=1}^{m-1}$ and $(\overline q_{j}(\delta))_{j=1}^{m'-1}$
be the constants that are defined at the end of Section~\ref{sec g}.
Let $\overline\sigma_\delta^*$ be the sender's message strategy defined in the statement of Lemma~\ref{lemma:structure} with these constants.

By Eqs.~\eqref{eq der left} and~\eqref{eq der right},
the function $\overline v_\delta$ is a solution of the following piecewise linear differential equation:
\[ f'(p) = \left\{
\begin{array}{lll}
\mu\cdot\left(\frac{\overline u_\delta(p_{-j})-f(p_{-j})} {p_{-j}-p^*}\right), & \ \ \ & p \in [p_{-j},\overline{q}_{-j+1}(\delta)), 2 \leq j \leq m-1, p^* \in (p_0,p_1),\\ \\
\mu\cdot\left(\frac{\overline u_\delta(p_{-j})-f(p_{-j})} {p_{-j}-p^*}\right), & \ \ \ & p \in [p_{-j},\overline{q}_{-j+1}(\delta)), 1 \leq j \leq m-1, p^* = p_0,\\ \\
\mu\cdot\left(\frac{\overline u_\delta(p)-f(p)} {p-p^*}\right),
& \ \ \ & p \in [\overline{q}_{-j+1}(\delta),p_{-j+1}], 2 \leq j \leq m-1,\\ \\
\mu\cdot\left(\frac{\overline u_\delta(p)-f(p)} {p-p^*}\right),
& \ \ \ & p \in [p_j,\overline{q}_{j}(\delta)], 1 \leq j \leq m'-1,\\ \\
\mu\cdot\left(\frac{\overline u_\delta(p_{j+1})-f(p_{j+1})} {p_{j+1}-p^*}\right),
& \ \ \ & p \in [\overline{q}_{j}(\delta),p_{j+1}], 1 \leq j \leq m'-1.
\end{array}\right.
\]

By Lemma~\ref{lemma deriv}, 
%Lemmas~\ref{lemma g} and~\ref{lemma g at endpoints},
$\overline{w}_\delta(\cdot,\overline\sigma^*_\delta)$ is also a solution of this differential equation.
By Lemma~\ref{lemma:CRRV}, $\overline\sigma_\delta^*$ is optimal on $[p_0,p_1]$ (if $p^* \in (p_0,p_1)$) or at $p_0$ (if $p^*=p_0$),
and therefore
$\overline{w}_\delta(\cdot,\overline\sigma^*_\delta) = \overline v_\delta$ on $[p_0,p_1]$
(if $p^* \in (p_0,p_1)$) or on $[p_{-1},p_1]$ (if $p^* = p_0$).
By the existence and uniqueness theorem for ordinary differential equations, $\overline{w}_\delta(\cdot,\overline\sigma^*_\delta) = \overline v_\delta$ on $[0,1]$.

%\subsection{Proofs of Results in Section~\ref{sec belief}}

%\color{red}
%/// Do we need the proof below? It repeats in mathematical symbols the intuition provided within the text. ///
%\color{black}

%\begin{proof}[Proof of Lemma~\ref{lemma:linear}]
%Since at receiver's belief $p$ it is optimal for the sender to split the belief between $p'$ and $p''$, we have
%\[\widehat{v}(p)=\left(\frac{p''-p}{p''-p'}\right)\widehat{v}(p')+ \left(\frac{p-p'}{p'-p'}\right)\widehat{v}(p'').\]
%The putative value function $w(\cdot; \sigma^{split}_{p',p''})$ of the strategy $\sigma^{split}_{p',p''}$ that splits the belief between $p'$ and $p''$ at every $p \in [p',p'']$ is linear. 
%Hence, if the limit value function in the interval $[p',p'']$ is not linear, there necessarily exists $\widetilde p\in(p',p'')$ such that 
%$\widehat v(\widetilde p) > w(\widetilde p;\sigma^{split}_{p',p''})$.

%Suppose w.l.o.g.~that
%$p'<\widetilde{p}<p<p''$.
%Then at the belief $p$, splitting between $\widetilde p$ and $p''$ yields a payoff strictly higher than $\widehat v(p)$, 
%a contradiction.
%\end{proof}

\subsection{Proofs of Intermediate Results for Lemma~\ref{lemma:structure}}\label{sec rem proofs}

\subsubsection{Proof of Lemma~\ref{lemma g}}
\label{section:proof:lemma g}

We will need the following intermediate result.

\begin{lemma}
\label{lemma v greater u}
For every $j\in\{1,...,m'\}$ we have $\overline{v}_\delta(p_j)< \overline{u}_\delta(p_j)$.
\end{lemma}

\begin{proof}
Fix $j\in\{1,...,m'\}$.
Denote $\widetilde u := \cav(\overline u_\delta)$.
Since in particular $\widetilde u \geq \overline u_\delta$, it follows that 
the value function $\widetilde v$ of $G_{cont}(\widetilde u)$ satisfies $\widetilde v \geq \overline u_\delta$.
Assumption~\ref{assumption:1} implies that $\widetilde u(p_j) = \overline u_\delta(p_j)$.
The function $\widetilde u$ is continuous and concave,
hence by Corollary~4 in \cite{cardaliaguet_markov_2016},
the optimal sender's message strategy in $G_{cont}(\widetilde u)$ is to never reveal information.
Hence,
\[ \overline v_\delta(p_j) \leq \widetilde v(p_j) = \int_0^\infty re^{-rt} \widetilde u(p^t) \rmd t
< \widetilde u(p_j) = \overline u_\delta(p_j), \]
where the process $(p^t)_{t \geq 0}$ under the integral term is given that the initial belief is $p_j$ and that the sender reveals no information.
The strict inequality holds because $\widetilde u$ is strictly increasing, and $p^t$ decreasing in $t$.
The claim follows.
\end{proof}

%The result follows from \cite{cardaliaguet_markov_2016}. This corollary states that
%if $u$ is concave, $v_{cont}(p)=\int_0^\infty re^{-rt}u(p^t)$.

%The value $v_{cont}(p)$ above is obtained by revealing no information (at time $t$ the instantaneous payoff is $u(p^t)$), that is, one optimal strategy for a concave $u$ is to reveal no information.

%If we replace $\overline{u}_\delta$ with the concavification of $\overline{u}_\delta$, then we  increase the $\overline{u}_\delta$ and so we may only increase the value. Observe that this concavification is still increasing. Using the above claim, the value of the changed game at $p_j$ is lower than $\overline{u}_\delta(p_j)$ - the belief decreases and so are the instantaneous payoffs.  If at the changes game the sender cannot obtain more that $\overline{u}_\delta(p_j)$ then so is the case with the value of the original $\overline{u}_\delta$. 

\noindent{\it Proof of Lemma~\ref{lemma g}.}
Recall that 
\begin{equation}
\label{equ:10.4}
\overline g_\delta(p)=\mu\cdot\frac{\overline{v}_\delta(p)-\overline{u}_\delta(p)}{p^*-p}, \ \ \ 
\forall p \in [0,1] \setminus \{p^*\}.
\end{equation}

\noindent\textbf{Proof of (a):}

Fix $p \in (p_{-j-1}, p_{-j}-\delta)$ for $j \in\{0,\ldots,m-1\}$.
Then 
\[ \overline u_\delta(p) = h_{-j-1} < {\overline{w}_\delta}(p,\sigma^{slide}_{p,p^*}) \leq \overline v_\delta(p). \]
Thus, 
on the interval $(p_{-j-1},p_{-j}-\delta)$,
the numerator in Eq.~\eqref{equ:10.4} is positive and by Lemma~\ref{lemma:monotonicity:v} it is nondecreasing.
The denominator in Eq.~\eqref{equ:10.4} is positive 
on this interval
and 
decreasing, 
and therefore $\overline g_\delta$ is increasing.

\bigskip
\noindent\textbf{Proof of (b):}

The derivative of $\overline g_\delta$ is
\begin{equation} 
\label{equ:lemma:g:1}
\overline g'_\delta(p)=\mu\cdot\frac{(\overline v'_\delta(p)-\overline u'_\delta(p))(p^*-p)+ (\overline v_\delta(p)-\overline u_\delta(p))}{(p^*-p)^2}, \ \ \ 
\forall p \in (0,1) \setminus \{p^*\}. 
\end{equation}
Let $p \in (p_{-j}-\delta,p_{-j})$, 
where $j\in\{1,\ldots,m-1\}$ (if $p^*=p_0$)
or $j\in\{0,1,\ldots,m-1\}$ (if $p^* \in (p_0,p_1)$). 
We have
$\overline{u}_\delta(p)=\left(\frac{h_{-j}-h_{-j-1}}{\delta}\right)(p-p_{-j}) + h_{-j}$.
Hence, on this interval
$\overline{u}_\delta'(p)=\frac{h_{-j}-h_{-j-1}}{\delta} > 0$, which is large for a small $\delta$.
Since $\overline{v}_\delta$ is concave and, by (G.3) for $p=0$ 
the derivative $\overline{v}_\delta'$ is bounded, positive, and at most $\overline g_\delta$. The functions $\overline{u}_\delta$ and $\overline{v}_\delta$ are bounded as well. 
Therefore, provided $\delta$ is sufficiently small, 
$\overline g'_\delta(p) < 0$ for every $p \in (p_{-j}-\delta, p_{-j})$.

\bigskip
\noindent\textbf{Proof of (c), item (i):}

Suppose that 
$p^*=p_0$.
On the interval $(p_0-\delta,p_0)$ we have $\overline{u}_\delta(p)=\frac{h_0-h_{-1}}{\delta}\cdot(p-p_0)+h_0$.
Hence, on this interval,
\[ \overline g_\delta(p)=\mu\cdot\frac{\frac{h_0-h_{-1}}{\delta}\cdot(p-p_0)+h_0-\overline{v}_\delta(p)}{p-p_0}=\mu\cdot\left(\frac{h_0-h_{-1}}{\delta} + \frac{h_0-\overline{v}_\delta(p)}{p-p_0}\right). \]
By Lemma~\ref{lemma:CRRV}, $\overline{v}_\delta(p_0)=h_0$. 
Therefore, $\frac{h_0-\overline{v}_\delta(p)}{p-p_0}=-\frac{\overline{v}_\delta(p_0)-\overline{v}_\delta(p)}{p_0-p}$. 
The concavity of $\overline{v}_\delta$ implies that 
$\overline g_\delta$ is increasing on this interval.

\bigskip
\noindent\textbf{Proof of (c), item (ii):}
Suppose again that 
$p^*=p_0$.

We need to show that for $\eta\in[0,p_1-p_0-\delta]$,

\[\frac{\overline{u}_\delta(p_0+\eta)-\overline{v}_\delta(p_0+\eta)}{\eta}<\frac{\overline{u}_\delta(p_1)-\overline{v}_\delta(p_1)}{p_1-p_0}.\]

By definition, $\overline{u}_\delta(p_0+\eta)=h_0$. 

Let $\widehat{\sigma}$ be a strategy that splits the receiver's belief between $p_0$ and $p_1$ for all beliefs in $[p_0,p_1]$. Observe that ${\overline{w}_\delta}(p_0+\eta,\widehat \sigma)$ does not depend on the strategy for beliefs outside $[p_0,p_1]$.
Plainly, $\overline{v}_\delta(p_0+\eta)
\geq \overline{w}_\delta(p_0+\eta,\widehat \sigma)$.
By Conclusion~\ref{conc: linear}, for $p=p_0+\eta$, $p'=p_0$, and $p''=p_1$ we have

\begin{eqnarray*}
&&\overline{w}_\delta(p_0+\eta, \widehat \sigma)\\
&&=h_0\cdot \frac{p_1(\mu+1)-p_0}{(p_1-p_0)(\mu+1)}
+ h_1\cdot \frac{p_0-p_0(\mu+1)}{(p_1-p_0)(\mu+1)}
+ (p_0+\eta)\cdot \mu\cdot \frac{h_1-h_0}{(p_1-p_0)(\mu+1)}\\
&&=h_0\cdot \frac{p_1(\mu+1)-p_0-p_0\mu-\eta\mu}{(p_1-p_0)(\mu+1)}
+h_1\cdot \frac{-h_0\mu+p_0\mu+\eta\mu}{(p_1-p_0)(\mu+1)}\\
&&=h_0\cdot \frac{(p_1-p_0)(\mu+1)-\eta\mu}{(p_1-p_0)(\mu+1)} + h_1\cdot \frac{\eta\mu}{(p_1-p_0)(\mu+1)}.
\end{eqnarray*}
Hence, 

\[\frac{h_0-\overline{v}_\delta(p_0+\eta)}{\eta}\leq \frac{h_0-\overline{w}_\delta(p_0+\eta, \widehat{\sigma})}{\eta}= \frac{\eta\mu(h_0-h_1)}{\eta(\mu+1)(p_1-p_0)}=\frac{\mu(h_0-h_1)}{(\mu+1)(p_1-p_0)}.\]

It is therefore sufficient to show that 

\[\frac{\mu(h_0-h_1)}{(\mu+1)(p_1-p_0)}<\frac{h_1-\overline{v}_\delta(p_1)}{p_1-p_0}.\]

Cancelling out the term $(p_1-p_0)$ and rearranging the remaining terms,
we see that it is sufficient to show that

\[h_1+\frac{\mu(h_1-h_0)}{(\mu+1)}>\overline{v}_\delta(p_1), \]
which holds by Lemma~\ref{lemma v greater u}.

\bigskip
\noindent\textbf{Proof of (d)}:

Let $j\in\{1,\ldots, m'\}$.
By Lemma~\ref{lemma v greater u}, $\overline{g}_\delta(p_j)>0$. Moreover, $\overline{u}_\delta$ and $\overline{v}_\delta$ are continuous, hence so is $\overline{g}_\delta$ on $[p_j,p_{j+1}-\delta]$. Therefore, there exists $d_j>0$ such that $\overline{g}_\delta$ is positive on $[p_j,p_j+d_j)$. 
In particular,
$\overline{u}_\delta>\overline{v}_\delta$ on $[p_j,p_j+d_j)$.

To see that $\overline{g}_\delta$ is decreasing on $[p_j,p_j+d_j)$, consider its derivative, given in Eq.~\eqref{equ:lemma:g:1}.
On $[p_j,p_{j+1}-\delta]$ we have $\overline{u}'_\delta=0$. 
By Lemma~\ref{lemma:monotonicity:v},
on this interval $\overline v_\delta$ is increasing, 
hence $\overline v'_\delta>0$. 
Since $p^*<p$, 
it follows that
$(\overline v'_\delta(p)-\overline u'_\delta(p))(p^*-p)<0$ on $[p_j,p_j+d_j)$ and  $\overline{u}_\delta>\overline{v}_\delta$, and therefore $\overline{g}'_\delta$ is negative on $[p_j,p_j+d_j)$.

Suppose that there is $d_j \in (0,p_{j+1}-p_j-\delta)$ such that 
$\overline{g}_\delta(p_j+d_j)=0$. 
As above, $\overline{g}'_\delta(p_j+d_j)<0$,
and therefore $\overline{g}_\delta$ keeps decreasing. 
Since $\overline{u}_\delta$ is constant on $[p_j,p_{j+1})$ and $\overline{v}_\delta$ is increasing on this interval, $\overline{g}_\delta$ remains negative on $[p_j+d_j,p_{j+1}-\delta]$.

\bigskip
\noindent\textbf{Proof of (e)}:

The proof is similar to the proof of item~(b).

\subsubsection{Proof of Lemma~\ref{lemma g at endpoints}}\label{sec proof g endpoints}

\noindent\textbf{Proof of (a):}

Fix $j\in\{2,\ldots,m\}$ (if $p^* = p_0$), or $j\in\{1,\ldots,m\}$ (if $p^* \in (p_0,p_1)$). 
We need to show that
\[\frac{\overline v_\delta(p_{-j})-\overline u_\delta(p_{-j})}{p^*-p_{-j}} > \frac{\overline v_\delta(p_{-j+1})-\overline u_\delta(p_{-j+1})}{p^*-p_{-j+1}},\]
or, equivalently,
\begin{equation}
    \label{equ:10.1}
\overline v_\delta(p_{-j})>\overline u_\delta(p_{-j})+\frac{p^*-p_{-j}}{p^*-p_{-j+1}}\cdot \left(\overline v_\delta(p_{-j+1})-\overline u_\delta(p_{-j+1})\right).
\end{equation}

Let $\sigma$ be a strategy that splits the receiver's belief between $p_{-j}$ and $p_{-j+1}$ for all beliefs in $(p_{-j},p_{-j+1})$, and, once the belief becomes $p_{-j+1}$, continues optimally (that is, ${\overline{w}_\delta}(p_{-j+1},\sigma)=\overline{v}_\delta(p_{-j+1})$). By Eqs.~\eqref{equ:time}
and~\eqref{equ:split1} we have
\begin{eqnarray}
\nonumber
\overline v_\delta(p_{-j})
&\geq& {\overline{w}_\delta}(p_{-j},\sigma)\\
    \label{equ:10.2}
&\geq& \frac{\mu\cdot(p_{-j+1}-p_{-j})}{p^*-p_{-j} + \mu\cdot(p_{-j+1}-p_{-j})}\cdot \overline u_\delta(p_{-j})\\
&&+ \left(1-\frac{\mu\cdot(p_{-j+1}-p_{-j})}{p^*-p_{-j}+ \mu\cdot(p_{-j+1}-p_{-j})}\right)\cdot  \overline v_\delta(p_{-j+1}).
\nonumber
\end{eqnarray}
Eqs.~\eqref{equ:10.1} and~\eqref{equ:10.2} imply that it is sufficient to show that
\begin{eqnarray*}
&&\frac{\mu\cdot(p_{-j+1}-p_{-j})}{p^*-p_{-j} + \mu\cdot(p_{-j+1}-p_{-j})}\cdot \overline{u}_\delta(p_{-j})+ \left(1-\frac{\mu\cdot(p_{-j+1}-p_{-j})}{p^*-p_{-j} + \mu\cdot(p_{-j+1}-p_{-j})}\right)\cdot \overline v_\delta(p_{-j+1}) \\ &&>u(p_{-j})+\frac{p^*-p_{-j}}{p^*-p_{-j+1}}\left(\overline v_\delta(p_{-j+1})-\overline{u}_\delta(p_{-j+1})\right).
\end{eqnarray*}
Simple algebraic manipulations show that
this inequality is equivalent to
%\begin{eqnarray*} 
%&&\frac{p^*-p_{-j}}{p^*-p_{-j} + \mu\cdot(p_{-j+1}-p_{-j})}\cdot \overline v_\delta(p_{-j+1})\\
%&&> \frac{p^*-p_{-j}}{p^*-p_{-j} + \mu\cdot(p_{-j+1}-p_{-j})}\cdot \overline{u}_\delta(p_{-j})+\frac{p^*-p_{-j}}{p^*-p_{-j+1}}\left(\overline v_\delta(p_{-j+1})-\overline{u}_\delta(p_{-j+1})\right).
%\end{eqnarray*}
%Cancelling out the term $p^* - p_{-j} > 0$, this inequality reduces to
%\begin{equation*} \frac{\overline v_\delta(p_{-j+1})}{p^*-p_{-j} + \mu\cdot(p_{-j+1}-p_{-j})} >  \frac{\overline{u}_\delta(p_{-j})}{p^*-p_{-j} + \mu\cdot(p_{-j+1}-p_{-j})} +\frac{\overline v_\delta(p_{-j+1})-\overline{u}_\delta(p_{-j+1})}{p^*-p_{-j+1}},\end{equation*}
%which solves to
%\begin{eqnarray*} \overline v_\delta(p_{-j+1})(p^*-p_{-j+1}) &>&  \overline{u}_\delta(p_{-j})(p^*-p_{-j+1})\\
%&&+(\overline v_\delta(p_{-j+1})-\overline{u}_\delta(p_{-j+1}))(p^*-p_{-j} + \mu\cdot(p_{-j+1}-p_{-j})).
%\end{eqnarray*}
%Further manipulations yield
%\begin{equation*} 
%\overline{u}_\delta(p_{-j+1})(p^*-p_{-j} + \mu\cdot(p_{-j+1}-p_{-j}))- \overline{u}_\delta(p_{-j})(p^*-p_{-j+1})>   \overline v_\delta(p_{-j+1})(\mu+1)(p_{-j+1}-p_{-j}).\end{equation*}
%This inequality in turn reduces to
\begin{equation}\label{eq p-j} \overline{u}_\delta(p_{-j+1})\cdot\frac{p^*-p_{-j} + \mu\cdot(p_{-j+1}-p_{-j})}{(\mu+1)(p_{-j+1}-p_{-j})}- \overline{u}_\delta(p_{-j})\cdot\frac{p^*-p_{-j+1}}{(\mu+1)(p_{-j+1}-p_{-j})}>   \overline v_\delta(p_{-j+1}).\end{equation}
Recall that $\overline{u}_\delta(p_{-j})=h_{-j}$ and $\overline{u}_\delta(p_{-j+1})=h_{-j+1}$.
To prove that Eq.~\eqref{eq p-j} holds, we will use a geometric argument rather than a long list of mathematical derivations.
Consider the line $\ell$ that passes through the points $(p_{-j},h_{-j})$ and $(p_{-j+1},h_{-j+1})$,
see Figure~\ref{fig:ell}.
Since $u$ has a concave envelope, the graph of $u$ lies below $\ell$, except at $p_{-j}$ and $p_{-j+1}$. Denote by $\widetilde{h}_1$ the unique real number such that $(p_1,\widetilde{h}_1)$ lie on $\ell$. Then $\widetilde{h}_1>h_1$. 

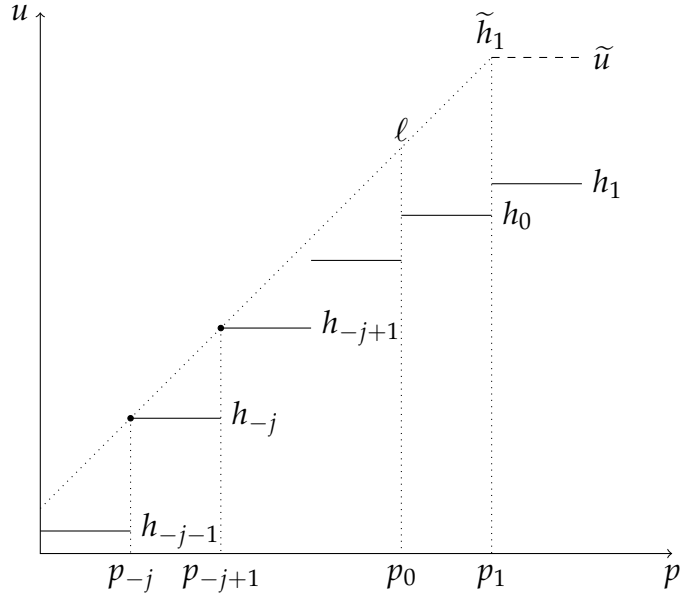
\begin{figure}[h]
\centering
\begin{tikzpicture}[domain=0:1,xscale=12,yscale=12]
\draw[<->] (0,.6) node[left]{$u$}-- (0,0) -- (.7,0) node[below] {$p$};
\draw[] (0,0.025)--(0.1,0.025) node[right]{$h_{-j-1}$};
\draw[] (0.1,0.15)--(0.2,0.15) node[right]{$h_{-j}$};
\draw[] (0.2,.25)--(0.3,.25) node[right]{$h_{-j+1}$};
\draw[] (0.3,.325)--(0.4,.325);
\draw[] (0.4,.375)--(0.5,.375) node[right]{$h_{0}$};
\draw[] (0.5,.41)--(0.6,.41) node[right]{$h_{1}$};
\draw[dotted] (0.1, 0.15)--(0.1,0) node[below]{$p_{-j}$};
\draw[dotted] (0.2,.25)--(0.2,0) node[below]{$p_{-j+1}$};
\draw[dotted] (0.4,.45)--(0.4,0) node[below]{$p_{0}$};
\draw[dotted] (0.5,.55)--(0.5,0) node[below]{$p_{1}$};
\draw[dotted] (0.0,0.05)--(0.5,.55);
\draw[dashed] (0.5,.55)--(0.6,.55) node[right]{$\widetilde u$};
\node at (0.5,.58) {$\widetilde h_{1}$};
\node at (0.4,.47) {$\ell$};
\draw [fill] (.2,.25) circle [radius=0.003];
\draw [fill] (.1,.15) circle [radius=0.003];
\end{tikzpicture}
\caption{The line $\ell$ and the function $\widetilde u$ in Case (a).} \label{fig:ell}
\end{figure}

Let $\widetilde u_\delta : [0,1] \to \dR$ be the function that coincides with $\overline u_\delta$ except on $[p_1,p_2-\delta)$ where it is equal to $\widetilde h_1$, and for $p\in [p_2-\delta, p_2]$, where it is equal to $\tfrac{h_2-\widetilde h_1}{\delta}\cdot(p-p_2)+h_2$.
%Consider an auxiliary problem, denoted $\widetilde{G}_{aux}$, 
%where the indirect payoff function, denoted $\widetilde{u}$, is the same as $u$ except on the interval $[p_1,p_2)$, where it is $\widetilde{h}_1$ (rather than $h_1$;
%in particular, $\widetilde u$ does not satisfy Assumption~\ref{assumption:1}). 
Restrict attention to beliefs in $[0,p_1]$.
Since the line $\ell$ lies above the graph of $u$, the concavification of $\widetilde{u}_\delta$ at $p^*$ is on $\ell$. 
%\color{red}The line $\ell$ is above $\widetilde{u}_\delta$ (defined as in Section~\ref{sec: u approx}) as well. Denote $\widetilde{G}_{aux, \delta}$ the game where the indirect payoff function is $\widetilde{u}_\delta$. \color{black}  
Lemma~\ref{lemma:CRRV} implies that on the interval $[p_{-j},p_1]$ the optimal message strategy $\widetilde\sigma^*$ in 
${G}_{cont}(\widetilde u_\delta)$ is to split the receiver's belief between $p_{-j}$ and $p_1$,
and the value function  $\widetilde v_\delta$ of ${G}_{cont}(\widetilde u_\delta)$ exists and is linear on this interval. 
We argue that on the interval $[p_{-j},p_1]$ we have

\begin{eqnarray}\nonumber
\widetilde{v}_{\delta}(p)&=&h_{-j}\cdot\frac{p_1\cdot(\mu+1) - p^*}{(p_1-p_{-j})(\mu+1)}+\widetilde{h}_1\cdot\frac{p^*-p_{-j}\cdot(\mu+1)}{(p_1-p_{-j})(\mu+1)}\\ \nonumber
&&+ p\mu\cdot\frac{\widetilde h_1-h_{-j}}{(p_1-p_{-j})(\mu+1)}.
\end{eqnarray}
By Conclusion~\ref{conc: linear}, 
%we can plug in the value of $\widetilde{h}_1$ and obtain:
\begin{eqnarray}
\label{equ:value:19}
\widetilde{v}_\delta(p_{-j+1})=\overline{u}_\delta(p_{-j+1})\cdot\frac{p^*-p_{-j} + \mu\cdot(p_{-j+1}-p_{-j})}{(\mu+1)(p_{-j+1}-p_{-j})}\nonumber\\
- \overline{u}_\delta(p_{-j})\cdot\frac{p^*-p_{-j+1}}{(\mu+1)(p_{-j+1}-p_{-j})}.
\end{eqnarray}
Since $\widetilde v_\delta \geq \overline v_\delta$, Eq.~\eqref{eq p-j} holds with  weak inequality.
When the initial belief $p$ is in the interval $[p_{-j+1},p_1]$, 
the only optimal strategy in $G_{cont}(\widetilde u_\delta)$ is the strategy 
that splits the receiver's belief between $p_{-j+1}$ and $p_1$.
Since $\widetilde h_1 > h_1$, this strategy yields in $G_{cont}(\overline u_\delta)$ 
a payoff lower than $\widetilde v_\delta(p)$.
It follows that $\widetilde v_\delta(p) > \overline v_\delta(p)$ for every $p < p_1$,
and the claim follows.

\bigskip

\noindent\textbf{Proof of (b):} 

Fix $j\in\{1,\ldots,m'-1\}$. 
We want to show that 

\[\frac{\overline{u}_\delta(p_j)-\overline v_\delta(p_j)}{p_j-p^*}>\frac{u(p_{j+1})-\overline v_\delta(p_{j+1})}{p_{j+1}-p^*}.\]

As for item~(a), it is sufficient to show that 
(compare this equation with Eq.~\eqref{eq p-j})

\begin{equation}\label{eq vpj small}
\overline v_\delta(p_j)<-\frac{p_j-p^*}{(\mu+1)(p_{j+1}-p_j)}\cdot \overline{u}_\delta(p_{j+1}) + \frac{p_{j+1}-p^* + \mu\cdot(p_{j+1}-p_j)}{(\mu+1)(p_{j+1}-p_j)}\cdot \overline{u}_\delta(p_j).
\end{equation}

As in item~(a), we consider an auxiliary problem.
Let $\ell$ be the line that passes through $(p_j,h_j)$ and $(p_{j+1},h_{j+1})$,
and let $\widetilde h_0$ be the unique real number such that $(p_0,\widetilde h_0)$ lies on $\ell$, so that $\widetilde{h}_0= \left(\frac{h_{j+1}-h_j}{p_{j+1}-p_j}\right)(p_0-p_j)+h_j$, see Figure~\ref{fig:ell1}.
Since $\overline{u}_\delta$ has a concave envelope, $\widetilde h_0 \geq h_0$.
Let $\widetilde u_\delta : [0,1] \to \dR$ be the function that coincides with $\overline{u}_\delta$,
except on $[p_0,p_1-\delta)$, where it is equal to $\widetilde h_0$.

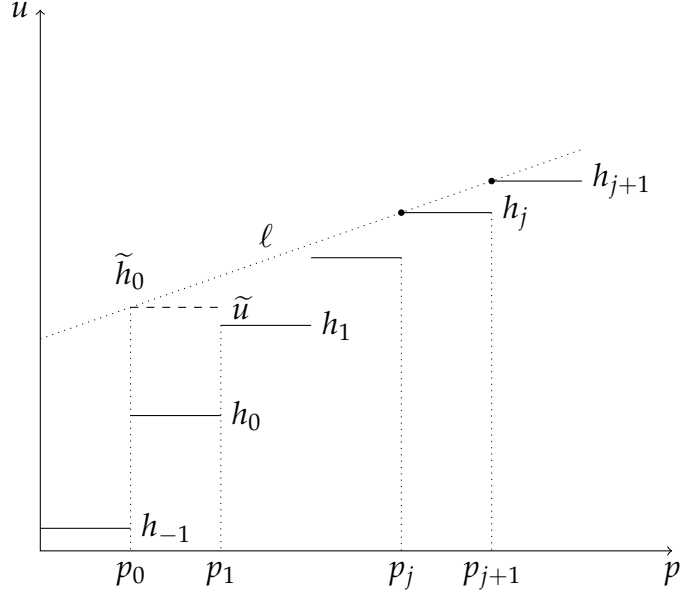
\begin{figure}[h]
\centering
\begin{tikzpicture}[domain=0:1,xscale=12,yscale=12]
\draw[<->] (0,.6) node[left]{$u$}-- (0,0) -- (.7,0) node[below] {$p$};
\draw[] (0,0.025)--(0.1,0.025) node[right]{$h_{-1}$};
\draw[] (0.1,0.15)--(0.2,0.15) node[right]{$h_{0}$};
\draw[] (0.2,.25)--(0.3,.25) node[right]{$h_{1}$};
\draw[] (0.3,.325)--(0.4,.325);
\draw[] (0.4,.375)--(0.5,.375) node[right]{$h_{j}$};
\draw[] (0.5,.41)--(0.6,.41) node[right]{$h_{j+1}$};
\draw[dotted] (0.1, 0.27)--(0.1,0) node[below]{$p_{0}$};
\draw[dotted] (0.2,.25)--(0.2,0) node[below]{$p_{1}$};
\draw[dotted] (0.4,.325)--(0.4,0) node[below]{$p_{j}$};
\draw[dotted] (0.5,.375)--(0.5,0) node[below]{$p_{j+1}$};
\draw[dotted] (0.0,0.235)--(0.6,.445);
\draw[dashed] (0.1,.27)--(0.2,.27) node[right]{$\widetilde u$};
\node at (0.1,.31) {$\widetilde h_{0}$};
\node at (0.25,.35) {$\ell$};
\draw [fill] (0.4,.375) circle [radius=0.003];
\draw [fill] (0.5,.41) circle [radius=0.003];
\end{tikzpicture}
\caption{The line $\ell$ and function $\widetilde u$ in Case~(b).} \label{fig:ell1}
\end{figure}

Assumption~\ref{assumption:1} implies that when the initial belief is in $[p_0,1]$, the value function of $G_{cont}(\widetilde u_\delta)$, 
denoted by $\widetilde v_\delta$,
is the line $\ell$.

By Lemma~\ref{lemma same der},

\begin{equation}
    \label{equ:19.1}
\widetilde{v}(p_j)=\widetilde{h}_0\cdot\frac{p_{j+1}(\mu+1)-p^*-\mu p_j}{(p_{j+1}-p_0)(\mu+1)} + h_{j+1}\cdot\frac{p^*-p_0(\mu+1) + \mu p_j}{(p_{j+1}-p_0)(\mu+1)}.
\end{equation}

Since $\widetilde u_\delta \geq \overline u_\delta$ 
on $[p_0,1]$,
we have 
$\widetilde v_\delta \geq \overline v_\delta$.
In particular, 
$\widetilde{v}_\delta(p_j) \geq \overline v_\delta(p_j)$.
Therefore, Eq.~\eqref{eq vpj small} will hold as soon as we show that 

\begin{equation}
    \label{equ:19.2}
\widetilde{v}(p_j)\leq -\frac{p_j-p^*}{(\mu+1)(p_{j+1}-p_j)}\cdot h_{j+1} + \frac{p_{j+1}-p^* + \mu\cdot(p_{j+1}-p_j)}{(\mu+1)(p_{j+1}-p_j)}\cdot h_j.
\end{equation}

Plugging the expression in Eq.~\eqref{equ:19.1} in Eq.~\eqref{equ:19.2},
and using the definition of $\widetilde{h}_0$, it is sufficient to show that

\begin{multline}
\left(\left(\frac{h_{j+1}-h_j}{p_{j+1}-p_j}\right)(p_0-p_j)+h_j\right)\frac{p_{j+1}\cdot(\mu+1)-p^*-\mu p_j}{(p_{j+1}-p_0)(\mu+1)}\\
+ h_{j+1}\cdot\frac{p^*-p_0\cdot(\mu+1) + \mu p_j}{(p_{j+1}-p_0)(\mu+1)}\\
\leq -\frac{p_j-p^*}{(\mu+1)(p_{j+1}-p_j)}\cdot h_{j+1} + \frac{p_{j+1}-p^* + \mu\cdot(p_{j+1}-p_j)}{(\mu+1)(p_{j+1}-p_j)}\cdot h_j.
\label{equ:19.3}
\end{multline}
Cancelling out the term $\mu+1$ and multiplying both sides of Eq.~\eqref{equ:19.3} by $(p_{j+1}-p_j)$, we see that we need to verify that
\begin{eqnarray}
\nonumber
&&\left((h_{j+1}-h_j)(p_0-p_j)+h_j\cdot(p_{j+1}-p_j)\right)\cdot \frac{p_{j+1}\cdot(\mu+1)-p^*-\mu p_j}{(p_{j+1}-p_0)} \\ 
\label{equ:34}
&&+h_{j+1}\cdot (p_{j+1}-p_j)\cdot \frac{p^*-p_0\cdot(\mu+1) + \mu p_j}{(p_{j+1}-p_0)} \\
&&\leq -(p_j-p^*)h_{j+1} + \bigl(p_{j+1}-p^* + \mu\cdot(p_{j+1}-p_j)\bigr)h_j.
\nonumber
\end{eqnarray}

The coefficients of $h_{j+1}$ in Eq.~\eqref{equ:34} cancel out, as do the coefficients of $h_j$.
Therefore Eq.~\eqref{equ:34} holds as an equality,
which implies that Eq.~\eqref{eq vpj small} holds with weak inequality. 
The proof that Eq.~\eqref{eq vpj small} holds with strict inequality uses the same arguments as for Part~(a).

\bigskip

\noindent\textbf{Proof of (c) items (i) and (ii) :}

These items are direct consequences of items (c) (i) and (c)(ii) of Lemma~\ref{lemma g}, respectively.

\subsection{Proofs of Lemmas~\ref{lemma:continuity}, \ref{lemma discrete ap}, and~\ref{lemma diminish dif}}\label{sec proofs three}

\subsubsection{Proof of Lemma~\ref{lemma:continuity}}\label{sec proof lemma conti}

The proof is by induction over the continuity intervals of $u$.

\bigskip
\noindent\textbf{Step 1:} The interval $[p_{0},p_1]$ when $p^* \in (p_0,p_1)$.

Suppose that $p^* \in (p_0,p_1)$.
On the interval $[p_0,p_1]$ 
the strategies $(\overline \sigma^*_\delta)_{\delta > 0}$ and $\sigma^*$ coincide:
they both split the receiver's belief between $p_0$ and $p_1$.
It follows that on this interval
$\overline v_\delta$ is independent of $\delta$,
and hence on $[p_0,p_1]$ 
\[ w(\cdot,\sigma^*) = w(\cdot,\overline \sigma^*_\delta) = \overline v_\delta = \overline v_0. \]

\bigskip
\noindent\textbf{Step 2:} The interval $[p_{-1},p_1]$ when $p^* = p_0$.

If $p^*=p_0$, then the strategies $\sigma^*$ and $(\overline \sigma^*_\delta)_{\delta > 0}$ coincide and instruct splitting the receiver's belief between $p_{-1}$ and $p^*=p_0$ 
(on the interval $[p_{-1},p_0]$ and between $p_0$ and $p_1$ (on the interval $[p_0,p_1]$).
The argument proceeds as in Step~1.

%If $p^* = p_0$, then at $p_0$    the strategies $(\overline \sigma^*_\delta)_{\delta > 0}$ and $\sigma^*$ coincide: the reveal no information, and $p^t = p_0$ for every $t \geq 0$.
%In particular, $w(p_0;\sigma^*) = w(p_0;\overline \sigma^*_\delta) = \overline v_0(p_0)=\overline{u}_\delta(p_0)=u(p_0)$. If $p^* = p_0$, then the strategies $(\overline \sigma^*_\delta)_{\delta > 0}$ and $\sigma^*$ coincide for $p\in(p_{-1}, p_1)$ as well: the strategies instruct splitting the belief between $p_{-1}$ and $p_0$ for beliefs in $[p_{-1},p_0]$, and similarly splitting between $p_0$ and $p_1$ for beliefs in $[p_0,p_1]$. The resulting values are convex combinations of $\overline{u}_\delta(p_{-1})=u(p_{-1})$ and $\overline{u}_\delta(p_0)=u(p_0)$ or convex combinations of  $\overline{u}_\delta(p_0)=u(p_0)$ and  $\overline{u}_\delta(p_1)=u(p_1)$, respectively, with the same weights of the same payoffs. We conclude that if $p^*=p_0$ then $w(p;\sigma^*) = w(p;\overline \sigma^*_\delta)= \overline v_0(p)$ for $p\in[p_{-1},p_1]$.

%In particular, $w(p;\sigma^*) = w(p;\overline \sigma^*_\delta) = \overline v_0(p)$, for every $p \in [p_0,p_1]$.

\bigskip
\noindent\textbf{Step 3:} The intervals to the left of $p^*$.

Suppose by induction that 
$\overline{v}_0(p)=w(p,\overline{\sigma}^*)$
for $p \in [p_{-j},p_0]$,
where $j \in \{1,2,\dots, {m-1}\}$ (if $p^* \in (p_0,p_1)$)
or $j \in \{2,\dots,m-1\}$ (if $p^*=p_0$).
Consider the interval $[p_{-j-1},p_{-j}]$.
On this interval, the strategy
$\sigma^*$ splits the receiver's belief between $p_{-j-1}$ and $p_{-j}$;
and for each $\delta > 0$ sufficiently small, the strategy $\overline \sigma^*_\delta$ splits the receiver's belief between $p_{-j-1}$ and $\overline q_{-j}(\delta)$,
and reveals no information between $\overline q_{-j}(\delta)$ and $p_{-j}$.
Since $\lim_{\delta \to 0}\overline q_{-j}(\delta) = p_{-j}$,
\[ \lim_{\delta \to 0} \overline w_\delta(\overline q_{-j}(\delta), \overline \sigma^*_\delta) = \lim_{\delta \to 0} \overline w_\delta(p_{-j}, \overline \sigma^*_\delta) = \overline v_0(p_{-j}) = w(p_{-j},\overline \sigma^*). \]
As a result, $\overline v_\delta$ converges to $w(\cdot,\overline\sigma^*)$ on $[p_{-j},p_{-j+1})$.

\bigskip
\noindent\textbf{Step 4:} The intervals to the right of $p^*$.

Suppose by induction that $\overline{v}_0(p)=w(p,\overline{\sigma}^*)$
for $p \in [p_{1},p_j]$,
where $j \in \{1,2,\dots,{m'-1}\}$,
and consider the interval $[p_j,p_{j+1}]$.
On this interval, the strategy
$\sigma^*$ reveals no information between $p_{j}$ and $q_j$,
and splits the receiver's belief between $q_{j}$ and $p_{j+1}$.
For each $\delta > 0$ sufficiently small,
the strategy $\overline \sigma^*_\delta$ 
reveal no information between $p_{j}$ and $\overline q_j(\delta)$,
and split the receiver's belief between $\overline q_{j}(\delta)$ and $p_{j+1}$.
Since $\lim_{\delta\to 0}\overline q_j(\delta)=q_j$,
the functions $\overline w_\delta(\cdot,\overline\sigma_\delta)$ converge to $w(\cdot,\overline\sigma^*)$ on $[p_j,q_j)$,
and by monotonicity,
the same holds at $q_j$.
As in Step~3, $\overline w_\delta(\cdot,\overline\sigma_\delta)$ converge to $w(\cdot,\overline\sigma^*)$ on $[q_j,p_{j+1}]$.

%For every $p\in[p_j,q_j)$ there exists a $\delta$ small enough so that\footnote{Over the subsequence of $\delta$, if this is the way the limit $q$ was taken.} $\overline{\sigma}^*_delta$ instructs not to reveal information, as does $\sigma^*$. Therefore, according to $\overline{\sigma}^*_\delta$, we obtain a convex combination of $h_j$ and $\overline{v}_\delta(p_j)$ while according to $\sigma^*$ we get a combination (with the same weights) of $h_j$ and $w(p_j,\simga^*)$. The inductive assumption shows convergence as claimed. 
%For $p\in[\overline q_j(\delta),p_{j+1}]$, $\overline v_\delta(p)$ is a convex combination of $h_{j+1}$ and $\overline{v}_\delta(\overline{q}_j(\delta))$, while $w(p,\overline{\sigma}^*)$ is a similar convex combination of $h_{j+1}$ and $w(p_j,\sigma^*)$. From the inductive assumption the convergence is as claimed, which completes the proof.

\subsubsection{Proof of Lemma~\ref{lemma discrete ap}}\label{sec proof lemma lim}

%We want to show that
%$\lim_{\Delta\to 0}w_\Delta(p,\sigma^*)=v(p)$.

%We first show two results on the convergence of the putative value of the discrete model to the continuous model. 
%We then use these results to prove the lemma inductively over the continuity intervals.

Recall that $Y^{split}_{p',p''}$ and $Y^{slide}_{p',p''}$
are the expected discounted time to reach belief $p''$ when the initial belief is $p'$ under $\sigma^{split}_{p',p''}$ and $\sigma^{slide}_{p',p''}$, respectively, in the  continuous-time game.
Denote by $Y^{\Delta,split}_{p',p''}$ and $Y^{\Delta,slide}_{p',p''}$
the corresponding quantities in the discrete-time game with length of period $\Delta$.
The reader can verify that for every distinct $p',p'' \in [0,1]$,
\[ \lim_{\Delta \to 0} Y^{\Delta,split}_{p',p''} = Y^{split}_{p',p''}, \]
and for every $p' < p'' < p^*$ and every $p^* < p'' < p'$,
\[ \lim_{\Delta \to 0} Y^{\Delta,slide}_{p',p''} = Y^{slide}_{p',p''}. \]
The proof now follows similar arguments to those used in the proof of Lemma~\ref{lemma:continuity}.

\subsubsection{Proof of Lemma~\ref{lemma diminish dif}}\label{sec proof lemma dim}

Since $u \leq \overline u_\delta$,
for every $\Delta>0$ and every $\delta>0$ sufficiently small we have 
$v_{\Delta}(u)\leq  v_\Delta(\overline{u}_\delta)$.
Taking the limit as $\Delta$ goes to 0, we have
$\lim_{\Delta\to 0}v_\Delta(u)\leq \lim_{\Delta\to 0}v_\Delta(\overline{u}_\delta)$.
By Theorem~1 in \cite{cardaliaguet_markov_2016}, for every $\delta>0$ we have
$\lim_{\Delta\to 0}v_\Delta(u_\delta)=v(\overline{u}_\delta)$.
We conclude that for every $\delta>0$ sufficiently small, 
$\lim_{\Delta\to 0}v_{\Delta}(u)\leq v(\overline{u}_\delta)$.
Since this inequality holds for every  sufficiently small $\delta > 0$,
taking the limit as $\delta$ goes to 0 yields
$\lim_{\Delta\to 0}v_\Delta(u)\leq \lim_{\delta\to 0}v(\overline{u}_\delta)=v(u)$,
where the last equality follows from Eq.~\eqref{eq equal values}.

\bibliographystyle{aea}
\bibliography{bib.bib}

\end{document}